\documentclass[journal]{IEEEtran}
                                                          
\usepackage{enumerate}
\usepackage{amsmath,amssymb,amsthm} 

\newtheorem{defin}{Definition}
\newtheorem{theorem}{Theorem}

\newtheorem{lemma}{Lemma}

\newcommand{\be}{\begin{equation}}
\newcommand{\ee}{\end{equation}}
\newcommand{\ba}{\begin{array}}
\newcommand{\ea}{\end{array}}
\newcommand{\bea}{\begin{eqnarray}}
\newcommand{\eea}{\end{eqnarray}}

\newcommand{\tran}{^{\mbox{\scriptsize T}}}  

\newcommand{\vbar}{\raisebox{.17ex}{\rule{.04em}{1.35ex}}}
\newcommand{\vbarind}{\raisebox{.01ex}{\rule{.04em}{1.1ex}}}
\newcommand{\D}{\ifmmode {\rm I}\hspace{-.2em}{\rm D} \else ${\rm I}\hspace{-.2em}{\rm D}$ \fi}
\newcommand{\T}{\ifmmode {\rm I}\hspace{-.2em}{\rm T} \else ${\rm I}\hspace{-.2em}{\rm T}$ \fi}
\newcommand{\B}{\ifmmode {\rm I}\hspace{-.2em}{\rm B} \else \mbox{${\rm I}\hspace{-.2em}{\rm B}$} \fi}
\newcommand{\Hil}{\ifmmode {\rm I}\hspace{-.2em}{\rm H} \else \mbox{${\rm I}\hspace{-.2em}{\rm H}$} \fi}
\newcommand{\C}{\ifmmode \hspace{.2em}\vbar\hspace{-.31em}{\rm C} \else \mbox{$\hspace{.2em}\vbar\hspace{-.31em}{\rm C}$} \fi}
\newcommand{\Cind}{\ifmmode \hspace{.2em}\vbarind\hspace{-.25em}{\rm C} \else \mbox{$\hspace{.2em}\vbarind\hspace{-.25em}{\rm C}$} \fi}
\newcommand{\Q}{\ifmmode \hspace{.2em}\vbar\hspace{-.31em}{\rm Q} \else \mbox{$\hspace{.2em}\vbar\hspace{-.31em}{\rm Q}$} \fi}
\newcommand{\Z}{\ifmmode {\rm Z}\hspace{-.28em}{\rm Z} \else ${\rm Z}\hspace{-.38em}{\rm Z}$ \fi}

\renewcommand{\vec}[1]{{\bf{#1}}}     

\usepackage{graphicx}
\usepackage{caption}
\usepackage{subcaption}
\usepackage{epstopdf}
\usepackage{hyperref}
\newcommand{\R}{\mathbb{R}}
\newcommand{\N}{\mathbb{N}}
\newcommand{\Blambda}{\lambda}

\newcommand{\BLambda}{\Lambda}

\newcommand{\BULambda}{\bar{\Blambda}}
\newcommand{\BLLambda}{\underline{\Blambda}}
\newcommand{\bla}{\boldsymbol\lambda}

\usepackage{nomencl}
\makenomenclature

\newcommand{\X}{X}
\newcommand{\Rm}{R}
\renewcommand{\v}{v}
\newcommand{\vb}{\bar{v}}
\newcommand{\vu}{\underline{v}}
\newcommand{\q}{q}
\newcommand{\qb}{\bar{q}}
\newcommand{\qu}{\underline{q}}
\newcommand{\pb}{\bar{p}}
\newcommand{\pu}{\underline{p}}
\newcommand{\p}{p}
\newcommand{\re}{^{\texttt{P}}}
\newcommand{\im}{^{\texttt{Q}}}

\usepackage{color,soul}
\definecolor{lightblue}{rgb}{.90,.95,1}
\sethlcolor{lightblue}
\usepackage{tikz}
\usetikzlibrary{shapes,arrows}

\newcommand{\un}[1]{\underline{#1}}

\newcommand{\dist}{\texttt{dist}}

\newtheorem{remark}{Remark}

\title{ \LARGE \bf Distributed  Optimal Voltage Control \\ with Asynchronous and Delayed   Communication 
}
\author{Sindri  Magn\'{u}sson, Guannan Qu, and Na Li
\thanks{*The work was supported by NSF 1608509, NSF CAREER 1553407, and ARPA-E through the NODES program, and the Harvard Climate Change Solution Funds.}
\thanks{Sindri  Magn\'{u}sson is with  the School of Electrical Engineering and Computer science at KTH  Royal  Institute  of  Technology. This work was partly performed while Sindri Magn\'{u}sson was a postdoc fellow at Harvard University. (Email: sindrim@kth.se)

Guannan Qu is with the Department of Computing and Mathematical Sciences at California Institute of Technology.  This work was partly performed while Guannan Qu was a PhD student at Harvard University.  (Email: gqu@caltech.edu)

Na Li is with the Harvard John A. Paulson School of Engineering and Applied Science  (Email:  nali@seas.harvard.edu)

}
}

\begin{document}

\maketitle

\begin{abstract}

  The increased penetration of volatile  renewable energy into distribution networks necessities  more  efficient distributed voltage control. In this paper, we design distributed feedback control algorithms where each bus can inject \emph{both active and reactive} power into the grid to regulate the voltages. The control law on each bus is only based on local voltage measurements and communication to its physical neighbors. Moreover, the  buses can perform their updates \emph{asynchronously}  without receiving information from their neighbors for periods of time. The algorithm enforces \emph{hard upper and lower limits} on the active and reactive  powers at every iteration. We prove that the algorithm converges to the optimal feasible voltage profile, assuming linear power flows. This provable convergence is maintained under bounded communication delays and asynchronous communications. We further numerically test the performance of the algorithm using the full \emph{nonlinear AC power flow} model. Our simulations show the effectiveness of our algorithm on realistic networks with both static and fluctuating loads, even in the presence of communication delays. 
  %
%
%
  \end{abstract}




 	\section*{{Nomenclature}}

\addcontentsline{toc}{section}{Nomenclature}

	\subsection{Parameters}
	
	\begin{IEEEdescription}[\IEEEusemathlabelsep\IEEEsetlabelwidth{$V_1,V_2,V_3$}]
	       \item [$N$] The number of buses (not including the substation) .
		\item [ $\mathcal{N}_0, \mathcal{N}, \mathcal{E}$] $\mathcal{N}=\{0,1,\ldots,n\}$ is the set of buses with $0$ being the substation; $\mathcal{N}=\mathcal{N}_0\setminus \{0\}$; $\mathcal{E}$ is the set of lines in the network.
		\item [$\sigma_i,\mathcal{C}_i,\mathcal{P}_i$]  $\sigma_i\in \mathcal{N}$ is the parent of bus $i$; 
		  $\mathcal{C}_i$ is the set of children of bus $i$;
		  $\mathcal{P}_i$ is the set of lines in the network on the unique path from the substation to bus $i$.
		\item [$\dist(i,j)$]  The  number of edges in the shortest path between the nodes $i,j\in \mathcal{N}$. 
		\item [$r_{ij}, x_{ij}$]  The resistance and reactance on the transmission line between $i,j$.
\item [$X$, $R$] The matrices in linearized branch-flow model.
		\item [$\bar{p}, \underline{p}$,$\bar{q}, \underline{q}$] Upper and lower limits of the active and reactive power, respectively.
		\item [$\bar{v}, \underline{v}, \bar{s}$] Upper voltage limit, lower voltage limits, and apparent power limits.
		\item [$a_i\re,b_i\re,c_i\re$] Parameters of the cost function $C\re(\cdot)$ related to bus $i\in \mathcal{N}$.
		\item [$a_i\im,b_i\im,c_i\im$] Parameters of the cost function $C\im(\cdot)$ related to bus $i\in \mathcal{N}$.
		\item [$a_{\min},L$] $a_{\min}=\{a_1\re,\ldots, a_N\re, a_1\im,\ldots, a_N\im\}$; $L=(||X||^2+||R||^2)/2$.
		\item [$t$] The iteration index.
		\item [$\tau_{ij}(t),\tau_{\max}$] $\tau_{ij}(t)$ is the communication delay in the communication link $(i,j)\in \mathcal{E}$ at iteration $t$; $\tau_{\max}$ is the maximum delay, i.e., $\tau_{ij}(t)\leq \tau_{\max}$ for all $(i,j)\in \mathcal{E}$ and $t$.
		\item [$\gamma$] The step-size in our algorithm
	\end{IEEEdescription}

	\subsection{Variables and Functions}
	\begin{IEEEdescription}[\IEEEusemathlabelsep\IEEEsetlabelwidth{$V_1,V_2,V_3$}]
		\item [$p,q,v$]  The active powers, reactive powers, and squared voltage magnitude, respectively.
		\item [$\bar{\lambda},\un{\lambda}$] Dual variables associated with the upper and lower voltage constraint, respectively. 
		\item  [$\lambda,\bla$] $\lambda=\un{\lambda}-\bar{\lambda}$; $\bla=(\un{\lambda},\bar{\lambda})$.
		\item [$\alpha$, $\beta\re$, $\beta\im$] The communicated messages.
		\item [$\hat{\alpha}$, $\hat{\beta}\re$, $\hat{\beta}\im$] Delayed version of the communicated messages.
		\item [$z_i\re,z_i\im$] $z_i\re$ is a local estimation of $R\lambda$;  $z_i\im$ is a local estimation of $X\lambda$.
	      \item [$C\re(\cdot), C\im(\cdot)$] The cost functions for the active and reactive powers, respectively. 
	      \item [$D(\cdot),\mathcal{L}(\cdot)$] The dual function and the Lagrangian function, respectively.

	\end{IEEEdescription}

\subsection{Notations}
	\begin{IEEEdescription}[\IEEEusemathlabelsep\IEEEsetlabelwidth{$\jmath:=\sqrt{-1}$}]
               \item  [$\R,\C,\N$] The set of real, complex, and natural numbers, respectively. 
               \item [$\R^n,\R^{n\times m}$] The set of real $n$ vectors and $n{\times} m$ matrices, respectively.
		\item[$ P_{ij},p_i$] The $i,j$-th entry of matrix $P$ and $i$-th entry of vector $p$, respectively.
		\item[$\langle \cdot,\cdot\rangle $] Inner product of vectors.
                \item [$\vec{i}$] The imaginary unit $\vec{i}=\sqrt{-1}$.
		\item[$\mathbf{1}$] $N\times1$ column vector with all ones. 	
		\item[$\lceil x\rceil+, \lbrack x\rbrack_{\un{x}}^{\bar{x}}$ ] The projection of vector $x$ onto the positive orthant and the box constraint $[\un{x},\bar{x}]$.
				\item[$\Vert \cdot \Vert$] Euclidean norm for vectors, spectral norm for matrices.

	\end{IEEEdescription}

%
%

 \section{Introduction}
\subsection{Motivation}
 Power girds  are increasing the volume of renewable energy generation from unpredictable sources such as solar and wind.  
 As a consequence, 
   large scale penetration of renewable energy will cause faster voltage fluctuations than today's networks can handle~\cite{carvalho2008distributed,molzahn2017survey}. 
   This means that too much  injection of renewable energy can easily overload the power systems.  
%
 However, the grid becomes better equipped to handle these challenges than before. 
 For example, many smart home appliances will have adjustable active power demands that can be used to stabilize the voltage fluctuations caused by abruptly changes in renewable power generation.  
 Similar flexible active power adjustments  can come from   smart distributed power generators and batteries of electric vehicles. 
  It is also possible to use flexible reactive power to regulate the voltage fluctuations, e.g., from PV-inverters. 
 However, to take advantage of the flexible active and reactive power injections and to use them to regulate the voltage fluctuations  sophisticated  control algorithms are needed.




\subsection{Related Work}
 There is a vast literature on voltage control algorithms.   
 Most works focus on VAR control where the buses regulate the voltage fluctuations by adjusting reactive power injection based on voltage measurements. 
%
 %
%
 %
 %
  Perhaps the most established of these algorithms are droop controllers~\cite{farivar2013equilibrium,jahangiri2013distributed},  which are  implemented in the
  IEEE 1547-2018 standard~\cite{droopStd2018ieee}.  
 In these algorithms  each  bus  updates its  reactive  power  based  on piecewise-linear control law from  local voltage measurements. 
 However, droop control  can fail in ensuring feasible voltages~\cite{Li2014} and can become inefficient in large networks~\cite{Zhu_2015}.
 Other algorithms based purely on local measurements have addressed some of these issues by relaxing voltage or reactive power constraints~\cite{Li2014,Zhu_2015}. 
%
  However,  even though such  local control algorithms may work well in some cases, e.g., when  the voltage or reactive power limits are relaxed, they generally fail in providing feasible solutions as   illustrated in~\cite{Cavraro2016,Bolognani2019}.
  In particular, they cannot guarantee that the voltage and reactive power limits are satisfied simultaneously. 
   This means that communication between the network's buses is necessary to solve the general voltage control problem. 
   
  This has motivated studies on distributed VAR voltage control where each bus updates its reactive power based on local voltage measurements and communications to its neighbors  in  the power network~\cite{Bolognani2013,bolognani2014distributed,magnusson2017voltage,kekatos2014stochastic,liu2018hybrid,liu2018distributed,qu2018optimal}.  
  The convergence of all these algorithms to a stable voltage profile   is proved under  linearzed power flow models.  
  However, for the algorithms in~\cite{Bolognani2013,bolognani2014distributed,magnusson2017voltage} to work the physical limits on the reactive power must be relaxed, which is often prohibitive in practice.  
  The algorithms in~\cite{kekatos2014stochastic,liu2018hybrid,liu2018distributed} relax the voltage constraint and introduce a penalty on violating the constraint instead. 
 Moreover, all of the above papers consider only reactive power control, whereas we consider both reactive and active power control in this paper.

 To perform joint active and reactive power control typically requires solving  an Optimal Power Flow (OPF) problem. 
 There is  much literature on distributed algorithms for solving OPF~\cite{dall2013distributed,erseghe2014distributed,magnusson2015distributed,vsulc2014optimal,Zhang_2015,kraning2014dynamic}.
 However, solving a full OPF problem is a time consuming process that requires multiple communication rounds.
 It is impractical to repeatedly solve a full OPF problem at the fast time scales that are needed to respond to volatile voltage fluctuations.  
 This has motivated studies on dynamic/online OPF algorithms~\cite{dall2016optimal,tang2017real}, where the OPF problem is updated at every iteration based on the most recent measurements. 
 However, these algorithms require global information, i.e.,  they assume that  at every iteration either a system operator communicates with all the buses  or that every bus communicates to every other bus. 
 Such global communications are often difficult or expensive since they lead  to long delays and large network congestion in addition to violating the privacy of buses. 
 Our work in this paper considers distributed algorithms  where only neighbors in the power network communicate and communication can be asynchronous or delayed. 
  
  Limited communication, such as  asynchronous updates and delays,  are common in practice but hard to handle in distributed algorithms. Most existing voltage control algorithms  require the buses to wait until they  receive  information from all of their neighbors before they can perform a control action,   which is clearly a limitation if the algorithms are supposed to run in real time.  There are some exception, however. 
   For example, algorithms with random package delays and asynchronous updates are studied in~\cite{gatsis2012residential,liu2018distributed,bolognani2014distributed}. Other types  of communication limitations have also been considered, such as event triggered communications   \cite{olivier2015active,fan2016distributed,Magnusson2019optimal} and limited bandwidth~\cite{magnusson2019voltage}. 
  However, all of these works have some limitations. For example, the algorithms in~\cite{gatsis2012residential,olivier2015active,Magnusson2019optimal} require global communications and the work   in~\cite{bolognani2014distributed,liu2018distributed,magnusson2017voltage,magnusson2019voltage} only considers reactive power control and must either relax the reactive power or voltage constraints to ensure convergence.  

\subsection{Main Contributions}

  The main contribution of this paper is to design asynchronous distributed algorithms for optimal voltage control using \emph{both active and reactive power} adjustment.  
  Each bus updates its active/reactive power with a local control law that is only based on local voltage measurements and communications  from its \emph{neighbors in the network}. 
  The buses can perform their local updates asynchronously even if they do not receive any communication from other nodes over some period of time. 
  The algorithm \emph{enforces hard upper and lower limits} on the active and  reactive powers at every iteration of the algorithm.  
   We prove the algorithms converges to an optimal solution to an optimal power flow problem with a feasible voltage profile, even with  asynchronous and delayed communications. We prove the convergence assuming a linearized relationship between voltage and  power injections.  However, we illustrate the performance of our algorithm using the \textit{full nonlinear AC power flow model} in the numerical studies.
   We show that our algorithm can well handle \emph{time-varying environments} where the loading situation of the distribution network is changing in the meantime of the algorithm.  
   Moreover, our numerical results shows  that our algorithm can reduce $80\%$ of the communication compared to a synchronous algorithm for achieving similar voltage control performance.

  There is an intuitive explanation for why our algorithm is robust to asynchronous and delayed communications. 
  Our algorithm is equivalent to asynchronous dual decomposition algorithms~\cite{Low_1999,chiang2007layering}.   This equivalence is not obvious. In fact, the major efforts of our proofs go into showing this equivalence.
  Nevertheless, this means that our algorithm enjoys the strong robustness properties for asynchronous communications that have been established over a long time for dual decomposition in theory and practice~\cite{chiang2007layering}.

It should be highlighted that our work makes a significant contribution  to distributed voltage control even in the absence of asynchronous and delayed communications. 
 This is, firstly,  because existing distributed voltage control algorithms consider only reactive power control. 
   Moreover, it is generally not possible to extend the ideas used to decompose these algorithms to handle both reactive and active power control.  
 Secondly, most distributed voltage control algorithms that ensure convergence to a feasible voltage profile do so  by allowing a violation of the reactive power constraint in the transient.
 There are two exceptions to this~\cite{Bolognani2019,qu2018optimal}. 
    In the algorithm in~\cite{Bolognani2019}, to compute each new control action the nodes need to solve a subproblem by performing multiple iterations of communications, which is clearly limiting for algorithms that should run in real-time.  
     However, our algorithm only requires one communication round per control action. 
   Compared to~\cite{qu2018optimal}, our  algorithm  development  is  different. 
   The algorithm in~\cite{qu2018optimal} is based on inexact primal-dual saddle point iterations, which generally converge very slowly.
    Our  algorithm  is  equivalent to   asynchronous  dual  decomposition, which generally has better convergence properties. This is also why we can prove convergence in the presence of asynchronous and delayed communications.

A very preliminary version of this work was presented in~[33]. 
 Compared to this paper,~[33] considers only reactive power control, omits most of the proofs and only presents very limited numerical tests. Including real power as control actions requires a significant amount of change in the algorithm and the analysis. All the numerical tests are new and many high-fidelity cases are tested and discussed. 
 Finally, almost the entire paper has been rewritten,  with much more detailed discussions on the main results and proofs explaining why the method works.

\section{System Model and Problem Formulation}
 \label{sec:SysModandPF}

 %
 \subsection{System Model: Branch Flow  for Radial Networks}

   Consider a radial power distribution network with $N+1$ buses represented by the set $\mathcal{N}_0=\{0\}\cup \mathcal{N}$, where $\mathcal{N}=\{1,\ldots, N\}$.  
%
  Bus $0$ is a feeder bus and the buses in $\mathcal{N}$ are branch buses. 
 Let $\mathcal{E}\subseteq \mathcal{N}_0\times \mathcal{N}_0$ denote the set of directed flow lines, so if $(i,j)\in \mathcal{E}$ then $i$ is the parent of $j$.
 For each $i$, let  $s_i=p_i+\vec{i}q_i \in \C$, $V_i\in \C$, and $v_i\in \R_+$ denote the complex power injection, complex voltage, and squared voltage magnitude, respectively,  at Bus $i$.
 For each $(i.j)\in \mathcal{E}$, let $S_{ij}=P_{ij}+\vec{i}Q_{ij}\in \C$, $I_{ij}\in \C$, and $z_{ij}=r_{ij}+\vec{i}x_{ij}\in \C$ denote the complex power flow, current, and impedance in the line from Bus $i$ to Bus $j$.
 The relationship between the variables can be expressed as~\cite{baran1989optimal,Baran1989},
\begin{subequations} \label{eq:LinBranchFlow}
 \begin{align}
   -p_i =& P_{\sigma_i i}-r_{\sigma_i i} l_{\sigma_i,i}-\sum_{k:(i,k)\in \mathcal{E}} P_{ik}, \hspace{0.8cm}i\in \mathcal{N}, \\
   -q_i =& Q_{\sigma_i i}-x_{\sigma_i i} l_{\sigma_i,i}- \sum_{k:(i,k)\in \mathcal{E}} Q_{ik}, \hspace{0.6cm}i\in \mathcal{N}, \\
    v_j-v_i =&  -2(r_{ij} P_{ij}+x_{ij}Q_{ij}) \notag
               + (r_{ij}^2+x_{ij}^2)l_{ij}, \\ & \hspace{4.8cm}(i,j) \in \mathcal{E}, \\
    l_{ij} =& \frac{P_{ij}^2+Q_{ij}^2}{v_i}  \hspace{3.3cm}(i,j)\in \mathcal{E},
 \end{align}
\end{subequations}
 where $\sigma_i$ is the parent of bus $i\in \mathcal{N}$, i.e., the unique $\sigma_i\in \mathcal{N}_0$ with $(\sigma_i,i)\in \mathcal{E}$, and $l_{ij}=|I_{ij}|^2$. 

 We develop our voltage control algorithm for the general nonlinear power flow in Equation~\eqref{eq:LinBranchFlow}. 
 However, we prove the convergence of the algorithm by  consider a linearied version of the above model. 
 %
In particular, we consider the linear Distflow approximation of the above equations,  which gives a good approximation in radial distribution networks~\cite{Baran1989}. 
 The  linear Distflow model is obtained by setting $l_{ij}=0$ 
 in which case Equation~\eqref{eq:LinBranchFlow} can be written as
 \begin{equation}\v= \Rm \p+ \X \q + \vec{1}v_0,  \label{eq:PhysicalRelationship-1}\end{equation}
  where 
$\v=[v_1,\ldots,v_N]\tran$, $\q=[q_1,\ldots,q_N]\tran$, $p=[p_1,\ldots,p_N]\tran$, 
 $$\X_{ij}=2 \hspace{-0.4cm} \sum_{(h,k)\in\mathcal{P}_i\cap \mathcal{P}_j}  \hspace{-0.4cm} x_{hk},~~\text{ and }~~\Rm_{ij}=2 \hspace{-0.4cm} \sum_{(h,k)\in\mathcal{P}_i\cap \mathcal{P}_j} \hspace{-0.4cm} r_{hk},$$
 where $\mathcal{P}_i\subseteq \mathcal{E}$ is the set of edges in the path from Bus 0 to Bus $i$.

 \subsection{Optimal Voltage Control} \label{subsec:VoltReg}

 The goal of this paper is to design distributed feedback control laws for the active and reactive powers that drive the voltages $\v$ to some feasible range ${v}\in [\vu, \vb]$.
   To that end,  we assume that the active and reactive power injections can be adjusted within some interval $p \in[\pu,\pb]$ and $\q \in[\qu,\qb]$.\footnote{The active power (and the reactive power similarly) can be decomposed into $p=p^{\text{Adj.}}+p^{\text{Con.}}$ where $p^{\text{Adj.}}$ is adjustable reactive power and $p^{\text{Con.}}$ is the fixed reactive power consumption.}
   The active power can typically be adjusted by demand response programs in smart home appliances, HVAC systems, vehicle charging stations, etc. 
   The reactive  power can be adjusted by PV-inverters. 
  For active and reactive power injection $p,q\in \R^N$ the resulting voltage $v(p,q)$ can be computed by solving Equation~\eqref{eq:LinBranchFlow}, i.e.,
  \begin{align}
      v(p,q)= \text{ Solution to Equation~\eqref{eq:LinBranchFlow} for given $p$ and $q$. }
  \end{align}
 Ideally, we wish to find the optimal active and reactive power:
\begin{equation} \label{MainProblem} 
\begin{aligned}
& \underset{p,q\in \R^N}{\text{minimize}}
& & C\re(p) +C\im(q) \\
& \text{subject to}
&&  \vu \leq \v(p,q)  \leq \vb  \\
&&& \pu \leq p \leq \pb, \\
&&& \qu \leq \q \leq \qb.
\end{aligned}
\end{equation}
 where
 \begin{equation}
  C\re(p)= \sum_{i=1}^NC_i\re(p)=\sum_{i=1}^N \frac{a_i\re}{2} p_i^2+b_i\re p_i +c_i\re \label{eq:c_i_p}
 \end{equation}
 is the generation cost for active power and
 \begin{equation}
 C\im(q)= \sum_{i=1}^N C_i\im(q) = \sum_{i=1}^N \frac{a_i\im}{2} q_i^2+b_i\im q_i +c_i\im \label{eq:c_i_q}
 \end{equation}
 is the generation cost of reactive power. 
 We provide the dual of~\eqref{MainProblem} in Section~\ref{Sec:AI}.
 Throughout the paper we assume that Problem~\eqref{MainProblem} is feasible and $a_i>0$ for all $i$. 
 Moreover, set
 $$a_{\min}:=\min\{a_1\re,\ldots,a_N\re,a_1\im,\ldots,a_N\im\}.$$
 %
 %
 The goal of this paper is to devise distributed algorithms that solve the problem that are robust to communication delays and use asynchronous update among devices. 

\subsection{Distributed Optimal Voltage Control}

 We will propose a distributed feedback control algorithm to solve~\eqref{MainProblem}. 
 Ideally, we would like  algorithms  that use only local information. 
 That is, each bus $i\in \mathcal{N}$ initializes its active and reactive powers as 
 $$ p_i(0)\in[\un{p}_i,\bar{p}_i] ~~\text{ and }~~   q_i(0)\in[\un{q}_i,\bar{q}_i]$$ 
 and then updates it as follows, for iteration index $t\in\N$,
  \begin{subequations} \label{EQ:FBC}
\begin{align} 
 & \textbf{Measurement:}&& v_i(t) = {v}_i(\q(t)) \label{EQ:FBC-A} \\
 &\textbf{P-Control:} &&    p_i(t{+}1){=}\texttt{AlgP}_i^t(\texttt{Local\_Info}_i(t)),  \label{EQ:FBC-B}   \\
 &\textbf{Q-Control:} &&    q_i(t{+}1){=}\texttt{AlgQ}_i^t(\texttt{Local\_Info}_i(t)),  \label{EQ:FBC-C} 
\end{align}
\end{subequations}
 where $\texttt{AlgP}_i^t:\R^{3(t+1)}\rightarrow [\un{p}_i,\bar{p}_i]$ and $\texttt{AlgQ}_i^t:\R^{3(t+1)}\rightarrow [\un{q}_i,\bar{q}_i]$ are, respectively, the local active and reactive power control algorithms and
 \begin{align*} 
    \texttt{Local\_Info}_i(t)=&\{p_i(0),\ldots,p_i(t),q_i(0),\ldots,q_i(t), \\ &~~~~~~~~~~~~~~~~~~~~~~~~~v_i(0),\ldots,v_i(t)\},
  \end{align*} 
is the local information available to bus $i$ at iteration $t$.   
   Unfortunately, there exists no  local algorithm that is guaranteed to solve the optimization problem, due to the impossibility result in~\cite{Cavraro2016}. 
 Therefore, it is necessary to include some communication into the control law.  
 Such communication can be modeled as follows:  
 \begin{align*} 
    (p_i(t{+}1),q_i(t{+}1)) = \texttt{Alg}_i^t(\texttt{Local\_Info}_i(t),\texttt{Comm}_i(t)),
 \end{align*}
 where $\texttt{Comm}_i(t)$ is information that bus $i$ has received from other buses until iteration $t$. 
 In this paper, we consider algorithms in this form 
 when the communicated information $\texttt{Comm}_i(t)$ at each iteration comes only from physical neighbours of node $i$.  
 Moreover, there the algorithms considered in this paper are provably robust  to asynchronous and delayed communication. 


\section{Algorithm and Main Results} 
 \label{subseq:PF-LDP}

 


 \subsection{Algorithm} \label{subseq:Alg}

 We now illustrate the distributed algorithm for solving Problem~\eqref{MainProblem}. 
 We first illustrate the main steps of the algorithm  and then provide the main convergence results. 
 
\noindent \rule{\columnwidth}{2.5pt}

\noindent \textbf{DIST-OPT:  Distributed  Optimal Voltage Control}

\vspace{-0.2cm}

\noindent \rule{\columnwidth}{2.5pt}

\begin{enumerate}[\bf STEP 1]
 
\item  \textbf{Initialization:} Set $t=0$,  $z_i\re(t)=z_i\im(t)=\un{\lambda}_i(0)=\bar{\lambda}_i(0)=\alpha_i(0)=\beta_i(0)=0$ for $i\in \mathcal{N}$.

\item \textbf{Local Control:} Each bus $i\in \mathcal{N}$ 
 injects into the grid the active and  reactive power 
  \begin{align*}
     p_i(t) = \left[\frac{1}{a_i\re }\left( z_i\re(t)-b_i\re \right) \right]_{\underline{p}_i}^{\bar{p}_i}\\
     q_i(t) = \left[\frac{1}{a_i\im }\left( z_i\im(t)-b_i\im \right) \right]_{\underline{q}_i}^{\bar{q}_i}     
  \end{align*}

  \item \textbf{Local Measurement:} Each bus $i\in \mathcal{N}$ measures the voltage magnitude 
 \begin{align*}
    v_i(t)=v_i(p(t),q(t)) 
 \end{align*} 
  and then updates 
%
  \begin{subequations} \label{eq:LamUp}
  \begin{align}
     \un{\lambda}_i(t+1) =& \lceil\un{\lambda}_i(t)+\gamma(\un{v}_i-v_i(t))\rceil_+ \label{eq:DLamUp1}\\ 
     \bar{\lambda}_i(t+1) =& \lceil \bar{\lambda}_i(t)+\gamma(v_i(t)-\bar{v}_i)\rceil_+ \label{eq:DLamUp2} \\
     \lambda_i(t+1) =&\un{\lambda}_i(t+1)-\bar{\lambda}_i(t+1),
  \end{align}
  \end{subequations}
 where $\gamma>0$ is a step-size parameter.
  
   \item \textbf{Communication:} Each bus $i\in \mathcal{N}$ sends the following information to its neighbours: 
\begin{itemize}
 \item   \textbf{If $i$ has a parent:} Send to parent $\sigma(i)$ the variable 
\begin{align}  \label{EQ:AlphaUpdate}
   \alpha_i(t+1)=\lambda_i(t+1)+\sum_{j\in \mathcal{C}_i} \hat{\alpha}_j(t), 
\end{align}
  where $\mathcal{C}_i$ is the set of the children of node $i$. The parent $\sigma(i)$ receives the possibly delayed version
  $$ \hat{\alpha}_i(t+1)=  \alpha_i(t+1-\tau_{i\sigma(i)}(t)). $$  
 \item   \textbf{If $i$ has a child:} Send to  each child $j\in \mathcal{C}_i$ the variable 
 $$ \beta_j(t+1)=(\beta_j\re(t+1),\beta_j\im(t+1))$$
 where\footnote{If node $i$ has no parent then set $\hat{\beta}\re_{i}(t)=\hat{\beta}\im_{i}(t)=0$.} 
\begin{align}
  \beta_j\re(t+1)= & R_{ii}\left(\lambda_i(t+1)+ \sum_{r\in \mathcal{C}_i\setminus \{j\}}   \hat{\alpha}_r(t) \right) \notag \\ 
                            & + \hat{\beta}_{i}\re(t). \\
  \beta_j\im(t+1)= & X_{ii}\left(\lambda_i(t+1)+ \sum_{r\in \mathcal{C}_i\setminus \{j\}}   \hat{\alpha}_r(t) \right) \notag \\ 
                            & + \hat{\beta}_{i}\im(t).  \label{EQ:BetaUpdate}
\end{align}
    Each child $j$ receives the possibly delayed version
   \begin{align*}
          \hat{\beta}_j\re(t+1)&=  \beta_j\re(t+1-\tau_{ij}(t)),\\
          \hat{\beta}_j\im(t+1)&=  \beta_j\im(t+1-\tau_{ij}(t)).
  \end{align*}
\end{itemize}

  \item \textbf{Local  Computation:} Each bus $i\in \mathcal{N}$ updates
  \begin{align} 
        z_i\re (t+1)=&R_{ii}\bigg(\lambda_i(t+1) \notag
        {+}\sum_{j\in \mathcal{C}_i} \hat\alpha_j(t+1) \bigg)  \\&+\hat\beta_{i}\re(t+1). \\
       z_i\im(t+1)=&X_{ii}\bigg(\lambda_i(t+1) \notag
        {+}\sum_{j\in \mathcal{C}_i} \hat\alpha_j(t+1) \bigg)  \\&+\hat\beta_{i}\im(t+1). \label{eq:MainZ}
  \end{align}

\item \textbf{Update Iteration Index:} $t=t+1$. 

\end{enumerate}

\noindent \rule{\columnwidth}{2.5pt}

 Note that all the variables besides voltage $v$ and the active/reactive powers $p$ and $q$ are axillary variables.
  As illustrated in the next section, they are related to  the dual variables of problem~\eqref{MainProblem}. 
 In \textbf{STEP 1} of the algorithm each bus initializes its parameters.  For simplicity of presentation, all parameters are initialized at $t=0$. 
 In \textbf{STEP 2} each bus $i$ injects active and reactive power into the system based on the available information in $z_i\re(t)$ and $z_i\im(t)$. 
 In \textbf{STEP 3} each bus $i$ takes a local measurement of the voltage of $v_i(p(t),q(t))$ corresponding to the  active and reactive power injections  $p(t)$ and $q(t)$.
 Moreover, based on these measurements  bus $i$ also updates the parameters $\un{\lambda}_i$, $\bar{\lambda}_i$, and ${\lambda}_i$ according to~\eqref{eq:LamUp}.\footnote{We show in Section~\ref{Sec:AI} that $\un{\lambda}_i$, $\bar{\lambda}_i$ are prices (or the dual variables) for violating the voltage constraint $\un{v}\leq v(q)\leq \bar{v}$.} 
  In \textbf{STEP 4} each bus $i$ communicates the parameter $\alpha_i(t+1)$ to their parent bus (cf. Equation~\eqref{EQ:AlphaUpdate}) and $\beta_j(t+1)$ to each of their children buses $j\in \mathcal{C}_i$. 
    In \textbf{STEP 5}, each bus $i$ updates its variable $z_i(t+1)$ based on the local information $\lambda_i(t+1)$ and $\alpha_j(t+1)$ received from each of its child's  $j\in \mathcal{C}_i$ and $\beta_i(t+1)$ received from its parent. 
    Note that the information in  $\alpha_j(t+1)$ and  $\beta_i(t+1)$ received by bus $i$ delayed by $\tau_{ji}$ and $\tau_{\sigma_ii}$, respectively.

   It should be highlighted the computation done at each iteration by each bus consists of only a few binary operations per iteration and takes only a few microseconds. 
  In particular, if we count  the number of binary operations performed at each  iteration then we find that 
  \begin{itemize}
    \item in \textbf{step 2)} each bus performs  $1$ subtraction, $1$ division, and $1$ projection,
    \item in \textbf{step 3)} each bus performs $2$ additions, $3$ subtractions, $2$ multiplications, and $2$ projections,
     \item in  \textbf{step 4)} each bus performs at most $3c+2$ additions and 2 multiplications, where $c$ is the maximal number of children of a node, i.e., $c=\max_{i=1,\ldots,N} |\mathcal{C}_i|$ where $\mathcal{C}_i$ is the set of children of node $i$ and $|\mathcal{C}|$ denotes the cardinality of the set $\mathcal{C}$,
      \item  in  \textbf{step 5)} each bus performs at most $2c+4$ additions and 2 multiplications.
  \end{itemize}
  That is at most $22+5c$ binary operations, where $c$ is typically small, e.g., $c=4$ for the test network we use in Section~\ref{Sec:Simulation}. 
  This  computation takes few microseconds on modern processors, as we report in Section~\ref{Sec:Simulation}.
  

\subsection{Main Results}
  We illustrate the performance of the algorithm on the full nonlinear power flow model in Section~\ref{Sec:Simulation}. 
  Due to the high nonlinearities of the AC power flows it is generally difficult to prove the convergence of voltage control algorithms using the  full AC model.
  However, we prove the algorithms convergence assuming the linear  relationship in Equation~\eqref{eq:PhysicalRelationship-1} between $v$, $p$, and $q$. 
  \begin{theorem}\label{MainTheorem}
   Suppose that 
   \begin{equation} \label{Eq:linearizedV}
      v(p,q)=Rp+Xq+\vec{1} v_0
   \end{equation}
   and that there exists $p,q\in\R^N$ such that (Slater's condition):
   \begin{equation} \label{Eq:Slater}
      \pu < p  < \pb,~~~~ \qu < \q  < \qb,~~ \text{ and } ~~\vu < \v(\q)  < \vb, 
   \end{equation}
   and that the communication delays are bounded by $\tau_{\max}$, i.e., $\tau_{ij}(t)\leq \tau_{\max}$ for all $i,j,t\in \N$.
   Let the  step-size $\gamma$ be chosen from the interval  
   \begin{equation}\label{eq:MainStepSize}
        \gamma\in \left(0, \frac{2 }{(1+((\tau_{\max}+1)d+1) \sqrt{N})L}   \right),
   \end{equation}
   where $d=\max_{i,j\in \mathcal{N}} \dist(i,j)$ is the diameter of the network and
   $$L=2\left(\frac{||R||^2+||X||^2}{a_{\min}}\right).$$
   Then the following holds
  $$\lim_{t\rightarrow \infty} (p(t),q(t))=(p^{\star},q^{\star}),$$
   where $(p^{\star},q^{\star})$ is the optimal solution to Problem~\eqref{MainProblem}.
 \end{theorem}
 \begin{proof}
  See Appendix~\ref{Sec:Conv}. 
 \end{proof}
 The theorem shows that our distributed algorithm converges to the optimal solution to the problem~\eqref{MainProblem} provided that the step-size is small enough. 
 Moreover, the convergence is maintained even if the communication from neighboring buses is delayed.  
 This means that each bus can asynchronously update its active and reactive powers based only on local measurements without waiting for communication from other nodes, as long as the time between communications is bounded.  
 We note that theoretical step-size in Equation~\eqref{eq:MainStepSize}. We show in the simulations in Section~\ref{Sec:Simulation} that much larger step-sizes can be used.
We now illustrate the intuition into why the algorithm works.

\subsection{Extension to General Cost and Constraint}\label{subsec:general}
	As 
will be	illustrated in the next section, \textbf{STEP 2} of the \textbf{DIST-OPT}  algorithm can be equivalently written as
	\begin{align}
	(p_i(t),q_i(t)) =& \arg\min_{p_i,q_i} C_i\re(p_i) +C_i\im(q_i)  - z_i\re(t) p_i -  z_i\im(t) q_i \nonumber\\
	\text{s.t.}\quad  &\un{p}_i\leq p_i\leq \bar{p}_i \label{step2_interpretation}\\
& \un{q}_i\leq q_i\leq \bar{q}_i \nonumber
	\end{align} 
	where here $C_i\re(\cdot)$ and $C_i\im(\cdot)$ are the quadratic functions defined in \eqref{MainProblem}, \eqref{eq:c_i_p}, and \eqref{eq:c_i_q}.
	In other words, $p_i(t)$ and $q_i(t)$ are in fact the solution of a local optimization problem, where the constraint corresponds to the local active/reactive power capacity constraint in \eqref{MainProblem}, and the cost corresponds to the local cost $C_i\re$ and $C_i\im$ in \eqref{MainProblem}, plus a linear term depending on the multipliers $z_i\re(t)$, $z_i\im(t)$. Here the multipliers $z_i\re(t)$, $z_i\im(t)$ capture the voltage constraint violation at time $t$, and their role in \eqref{step2_interpretation} can be understood as forcing the control action $(p_i(t), q_i(t))$ to respond to the voltage violation. As Theorem~\ref{MainTheorem} shows, the algorithm will converge to the solution of \eqref{MainProblem}.

	Given this interpretation, our algorithm can be extended to handle more general local cost functions and constraints. In particular, instead of the optimization problem \eqref{MainProblem}, we consider the following more general problem, 
	\begin{equation} \label{MainProblem_new} 
	\begin{aligned}
	& \underset{p,q\in \R^N}{\text{minimize}}
	& & \sum_{i=1}^N \tilde{C}_i(p_i,q_i) \\
	& \text{subject to}
	&&  \vu \leq \v(p,q)  \leq \vb  \\
	&&&  \mathbf{g}_i(p_i,q_i)\leq \mathbf{0} 
	\end{aligned}
	\end{equation}
	where  $\tilde{C}_i(p_i,q_i)$ is a strictly convex function and $\mathbf{g}_i(p_i,q_i)$ is convex. 
	Here the bold $\mathbf{g}_i$ means that it is a vector, and can include more than one constraint. As an example, this $\mathbf{g}_i$ could include the apparent power limit constraint,
	$$p_i^2 + q_i^2 \leq \bar{s}_{i}^2$$
	in addition to the box constraint $\underline{p}_i \leq p_i\leq\bar{p}_i$ and $\underline{q}_i\leq q_i \leq\bar{q}_i$.

	  In light of the interpretation \eqref{step2_interpretation}, we can actually derive a generalized version of our algorithm to solve the more generalized problem \eqref{MainProblem_new}. To do this, we simply replace  \textbf{STEP 2} in the \textbf{DIST-OPT}  algorithm with the following step, 
    \begin{align}
	  (p_i(t),q_i(t)) =& \arg\min_{p_i,q_i} \tilde{C}_i(p_i,q_i) - z_i\re(t) p_i - z_i\im(t) q_i \label{eq:general_constr} \\
	  \text{s.t.}\quad  & \mathbf{g}_i(p_i,q_i)\leq \mathbf{0}\nonumber
    \end{align} 
which is essentially \eqref{step2_interpretation} with the local cost and local constraint replaced with the generalized cost $\tilde{C}_i(p_i,q_i)$ and the generalized constraint $\mathbf{g}_i(p_i,q_i)\leq \mathbf{0}$.  Compared with the original  \textbf{STEP 2}, \eqref{eq:general_constr} is a local optimization problem that might not have a simple closed form solution as the original  \textbf{STEP 2}. However, in the case that $\tilde{C}_i$ are convex quadratic functions, and $\mathbf{g}_i$ is the apparent power constraint, \eqref{eq:general_constr} is a simple convex Quadratic Constrained Quadratic Programming (QCQP) problem with two variables and can be solved efficiently~\cite{convex_boyd}. In the simulation section, we test this more generalized form of our algorithm to verify its validity.

\section{Algorithm Intuition} \label{Sec:AI}



  We now give intuition into the algorithm and explain why it solves Problem~\eqref{MainProblem}.  
  %
  A key insight is that our \textbf{DIST-OPT} algorithm is equivalent to  asynchronous dual decomposition methods similar to~\cite{Low_1999,chiang2007layering} where the primal problem is solved using old dual variables.  
  This is in no way obvious. 
  For example, if we would directly apply similar dual decomposition approaches as in~\cite{chiang2007layering}  to~\eqref{MainProblem} then we get an algorithm where {every node needs to communicate  to every other node in the network.}  
%
     However, in our algorithm the nodes  communicate  only to their neighbours in the network. 
      We achieve this by introducing the axillary variables $z^R_i$, $z^Q$, $\beta^R$, $\beta^Q$, and $\alpha$, and their updates in \textbf{Step 4} and \textbf{Step 5} of the algorithm. 
      We designing these updates by exploiting the special structures of the voltage control problem.
    {These axillary variables and their updates are novel in our algorithm and make our algorithm fully distributed,  only neighbor to neighbor  communication is needed.} We illustrate this insight in more detail now.

We need to start by introducing the dual of the optimization problem in Equation~\eqref{MainProblem}, which is given by
\begin{equation} \label{dual_problem}
   \begin{aligned}
    & \underset{\bla=(\BLLambda,\BULambda)}{\text{maximize}}
    & & D(\bla) := \min_{(p,q) \in [\un{p},\bar{p}] \times [\un{q},\bar{q}]} \mathcal{L}(p,q,\bla) \\
    & \text{subject to}
    & &   \bla \in\R_+^{2N} ,
  \end{aligned}
\end{equation}
where $\BLLambda$ and $\BULambda$ are, respectively, the dual variable associated to the voltage lower and upper bounds and $D:\R^{2N}\rightarrow \R$ is the dual function and $\mathcal{L}(\cdot)$ is the Lagrangian function defined as 
 \begin{equation}\label{EQ:LAG}
      \mathcal{L}(p,q,\bla)=  C\re(p)+C\im(q) + \BLLambda\tran ( \vu-\v(p,q)) + \BULambda\tran ( \v(p,q)- \vb),
 \end{equation}
 where $\bla=(\un{\lambda},\bar{\lambda})\in \R^N\times \R^N$, see Chapter 5 in~\cite{nonlinear_bertsekas} for details.  
 We have the following result proved in Appendix~\ref{APP:Lemma:DualMain}.
 \begin{lemma} \label{Lemma:DualMain}
     The dual gradient is
\begin{align} \label{eq:dualGrad}
  \nabla  D(\bla) = \left[\begin{array}{c}   \vu-\v(p(\bla),q(\bla)) \\ \v(p(\bla),q(\bla) )-\vb  \end{array}\right]
\end{align}
 where   
  \begin{subequations}  \label{eq:LocalProblem-c}
  \begin{align}
  p(\bla)     =&    \left[ {\Lambda_{\texttt{P}}}^{-1}R( \BLLambda -\BULambda) -b\re \right]_{\pu}^{\pb}, \\ 
   \q(\bla)  =&    \left[ {\Lambda_{\texttt{Q}}}^{-1}{X}( \BLLambda -\BULambda) -b\im \right]_{\qu}^{\qb}, 
\end{align} 
\end{subequations}
  and 
  $$ \Lambda_{\texttt{P}}=\texttt{diag}(a_1\re,\ldots,a_N\re)~\text{ and }~\Lambda_{\texttt{Q}}=\texttt{diag}(a_1\im,\ldots,a_N\im).$$ 
  Moreover, $\nabla D(\bla)$ is $L$-Lipschitz continuous where 
  $$L=2\left(\frac{||R||^2+||X||^2}{a_{\min}}\right).$$
 \end{lemma}
 \begin{remark}
    Note that $(\p(\bla),q(\bla))$ in  Equation~\eqref{eq:LocalProblem-c} is the projection of the unconstrained minimizer of $L(p,q,\bla)$ to the box constraint $[\un{p},\bar{p}]\times[\un{q},\bar{q}]$. 
     In general, the optimal solution of a constrained optimization problem cannot be obtained by projecting the unconstrained optimizer to the constraint set, even for quadratic problems. 
    However, this works here because of the special structures of the matrices in our problem, see  Appendix~\ref{APP:Lemma:DualMain} for the details.
 \end{remark}
  From the lemma we can derive a standard dual decent  algorithm by setting  
 \begin{subequations}  \label{eq:primal_update_Cent}  
  \begin{align}
        p_i(t)=& \left[\frac{1}{a_i\re}\left( \sum_{j=1}^N R_{ij} \lambda_i(t)  -b_i\re \right)\right]_{\pu_i}^{\pb_i},  \\
        q_i(t)=& \left[\frac{1}{a_i\im}\left( \sum_{j=1}^N X_{ij} \lambda_i(t)  -b_i\im \right)\right]_{\un{q}_i}^{\bar{q}_i}, 
  \end{align}
 \end{subequations}    
  where $\lambda_i(t)= \BLLambda_i(t) -\BULambda_i(t)$ and $\bla=(\BLLambda_i,\BULambda_i)$ is updated according to Equation~\eqref{eq:LamUp}, which is equivalent to the following gradient update
  \begin{align} 
     \bla(t{+}1) =&   \lceil \bla(t)+ \gamma   \nabla D(\bla(t)) \rceil_+. \label{eq:dual descent}  
  \end{align} 
 This algorithm is guaranteed to converge to the optimal solution (provided that $\gamma>0$ is small enough) since it is simply gradient ascent  for maximizing the dual  $D(\cdot)$.  
 The downside of this update is that to calculate $q_i(t+1)$ in Equation~\eqref{eq:primal_update_Cent} bus $i$ needs information form every other bus in the network.  
 This is countered in the \textbf{DIST-OPT} algorithm where only neighbour to neighbour communication is needed.  
 This is obtained by replacing the sums in Equation~\eqref{eq:primal_update_Cent}  by 
 $$z_i\re (t)\approx \sum_{j=1}^N R_{ij} \lambda_i(t)~~\text{ and }~~ z_i\im (t)\approx \sum_{j=1}^N X_{ij} \lambda_i(t)$$
 which can be computed in distributed fashion, see Equation~\eqref{eq:MainZ}. 
 In particular,  $z_i\re(t)$ and $z_i\im(t)$ are delayed versions of $\sum_{j=1}^N R_{ij} (\lambda_i(t))$ and $\sum_{j=1}^N X_{ij} (\lambda_i(t))$, respectively, as shown in the following lemma
%
   (proved in the extended version~\cite{Magnusson2019Extended})
\begin{lemma} \label{lemma:delayEqu}
  If $\tau_{ij}(t)=0$ for all $i,j\in\mathcal{N}$ and $t\in \N$, then we have for all $i\in \mathcal{N}$ that
  
  \begin{subequations}  \label{Eq:Delay_Z}
  \begin{align}
        z_i\re(t)=&\sum_{j=1 }^N R_{ij} \lambda_j(t-d_{ij}) ~~\text{ and }\\
        z_i\im(t)=&\sum_{j=1 }^N X_{ij} \lambda_j(t-d_{ij})
  \end{align}
  \end{subequations}
  where $\lambda_j(t)=0$ for $t<0$ and
 \begin{equation}
      d_{ij} = \begin{cases}  0  & \text{if } \dist(i,j)\leq 1 \\
                                            \dist(i,j)-1 & \text{otherwise.}
                      \end{cases}
 \end{equation} 
 If $\tau_{ij}(t)\leq \tau_{\max}$ for all $i,j\in\mathcal{N}$ and $t\in \N$ then
   \begin{subequations}  \label{Eq:Delay_Z2}
   \begin{align} 
        z_i\re(t)=\sum_{j=1 }^N R_{ij} \lambda_j(t-\bar{\tau}_{ij}(t)), \\
        z_i\im(t)=\sum_{j=1 }^N X_{ij} \lambda_j(t-\bar{\tau}_{ij}(t)),         
  \end{align}
 \end{subequations}   
  where $\bar{\tau}_{ij}(t)\leq (\tau_{\max}+1)d$.
\end{lemma}
 The lemma shows that the \textbf{DIST-OPT} algorithm is equivalent to updating the dual variables $\bla=(\BLLambda_i,\BULambda_i)$ according to the recursion 
  \begin{align} 
    \bla(t{+}1) =&   \lceil \bla(t)+ \gamma   g(t) \rceil_+,\label{eq:delayed dual descent}  
 \end{align} 
 where $g(t)$ is an approximation of the dual gradient $\nabla D(\bla(t))$ using old $\lambda$ values solve~\eqref{eq:primal_update_Cent} (cf. Equation~\eqref{Eq:Delay_Z} and~\eqref{Eq:Delay_Z2}). In particular,
	\begin{align}
	(p(t),q(t)) =& \arg\min_{p,q} C\re(p) +C\im(q)  - z\re(t)\tran p -  z\im(t)\tran q \nonumber\\
	\text{s.t.}\quad  &\un{p}\leq p\leq \bar{p} \nonumber \\
& \un{q}\leq q\leq \bar{q} \nonumber
	\end{align}  
and  we have that
\begin{align} \label{eq:dualGradApprox}
    g(t) = \left[\begin{array}{c}   \vu-\v(p(t),q(t)) \\ \v(p(t),q(t) )-\vb  \end{array}\right].
\end{align}
 We  use this interpretation of the \textbf{DIST-OPT} algorithm to prove Theorem~\ref{MainTheorem}, see the extended version for details~\cite{Magnusson2019Extended}.

 \begin{figure}[!t]
	\centering
	\includegraphics[width=\columnwidth]{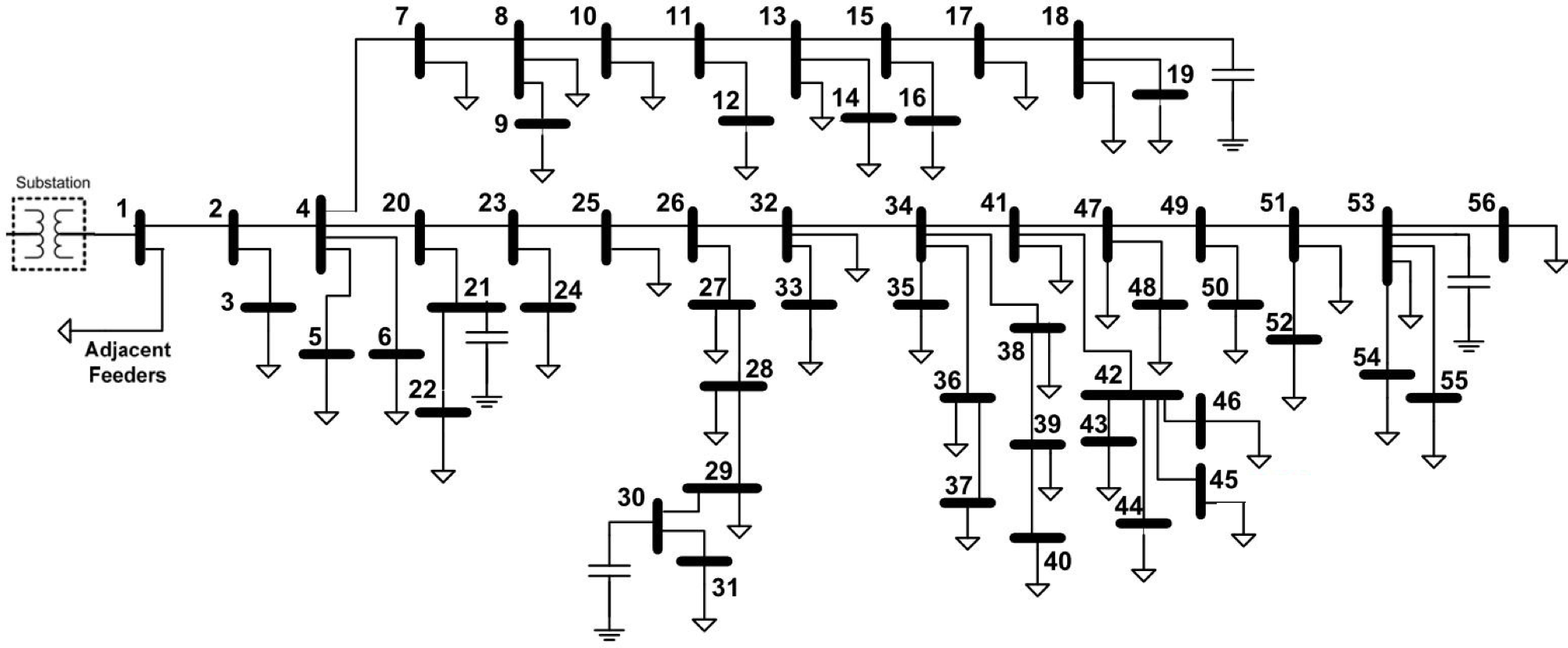}
	\caption{Schematic diagram of two SCE distribution systems. }
	\label{fig:circuit}
\end{figure}

   \begin{figure}
	\includegraphics[width=\columnwidth]{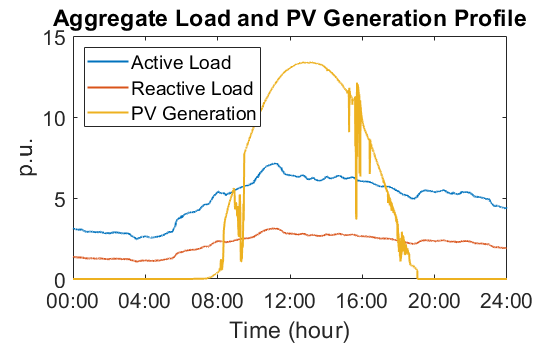}  
	\caption{Aggregated active load, reactive load and PV generation profile. } \label{fig:load}
\end{figure}

\section{Simulations} \label{Sec:Simulation}

\begin{figure*}
	\centering
	\begin{subfigure}[b]{0.32\textwidth}
		\includegraphics[width=\textwidth]{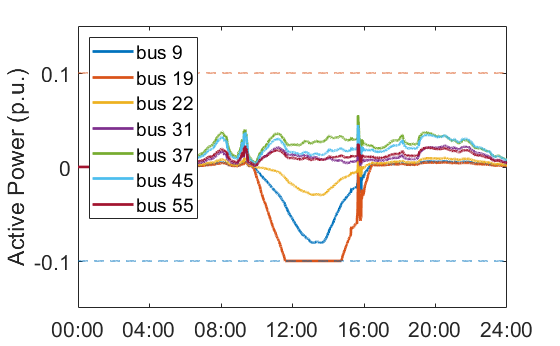}
		\caption{Active power injection.}
		\label{fig:dynamic_p}
	\end{subfigure}
	\begin{subfigure}[b]{0.32\textwidth}
		\includegraphics[width=\textwidth]{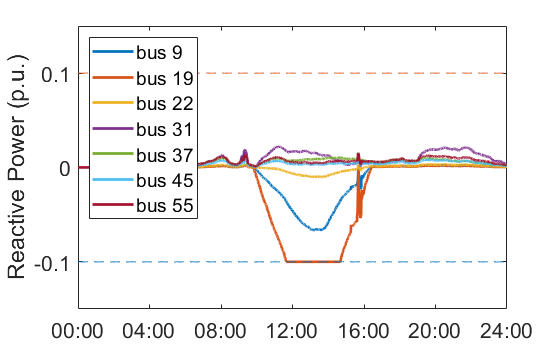}
		\caption{Reactive power injection.}
		\label{fig:dynamic_q}
	\end{subfigure}
	~
	\begin{subfigure}[b]{0.32\textwidth}
		\includegraphics[width=\textwidth]{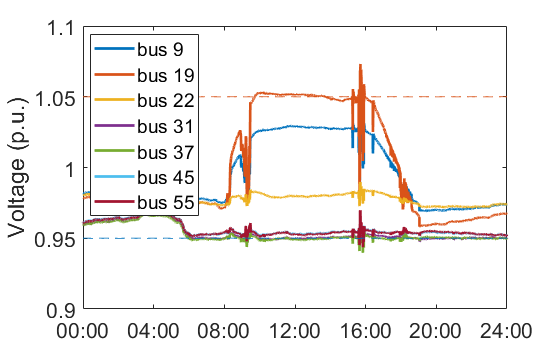}
		\caption{Voltage profile.}
		\label{fig:dynamic_v}
	\end{subfigure}
	\caption{Progress of the dynamic \textbf{DIST-OPT} algorithm when there are no time delays.}\label{fig:dynamic}
\end{figure*}


We evaluate our algorithm \textbf{DIST-OPT} on    the full nonlinear AC power flow model (\ref{eq:LinBranchFlow}), using Matpower \cite{5491276}. 
We do our experiments on the distribution circuit of South California Edison \cite{Farivar-2012-VVC-PES}.\footnote{See \cite{Farivar-2012-VVC-PES} for the network data including the line impedance, the peak MVA demand of the loads and the nameplate capacity of the shunt capacitors. }
All results are expressed as per-unit (p.u). 
The nominal voltage magnitude is $1 \text{p.u.}$ and the
 acceptable range is set as plus/minus 5\% of the nominate value. We divide the simulation into two parts. In the first part, we use a realistic load and PV generation data over a one-day period, and evaluate the performance of our algorithm.
 In particular, we evaluate its ability to keep voltage within the acceptable limits and its robustness against communication delays, measurement noise and modeling error. In the second part, we focus on the optimality of the proposed algorithm, i.e. how well the algorithm can minimize the optimization problem in Equation~\eqref{MainProblem}. Note that since now we are considering the full AC-nonlinear power flow in Equation~\eqref{eq:LinBranchFlow}, the optimization problem is nonconvex.

\subsection{Performance and Robustness under Time-Varying Load and PV Generation}\label{subsec:simu_timevarying}

\begin{figure*}
	\centering
	\begin{subfigure}[t]{0.32\textwidth}
		\includegraphics[width=\textwidth]{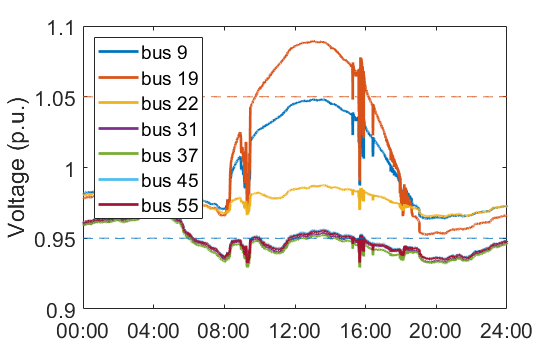}
		\caption{No control.}
		\label{fig:dynamic_v_withoutcontrol}
	\end{subfigure}
	\begin{subfigure}[t]{0.32\textwidth}
		\includegraphics[width=\textwidth]{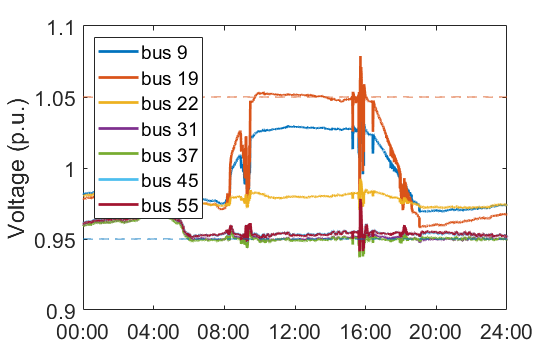}
		\caption{Fixed delay $\tau_{ij}(t)=5$.}
		\label{fig:robust_delay_fixed}
	\end{subfigure}
	~
	\begin{subfigure}[t]{0.32\textwidth}
		\includegraphics[width=\textwidth]{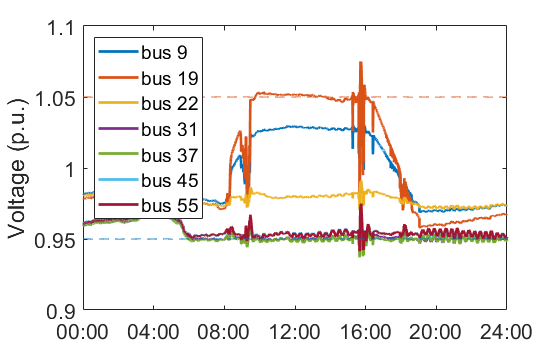}
		\caption{Random delay, $\tau_{\max}=15$.}
		\label{fig:robust_delay_random}
	\end{subfigure}
\\
	\begin{subfigure}[t]{0.32\textwidth}
		\includegraphics[width=\textwidth]{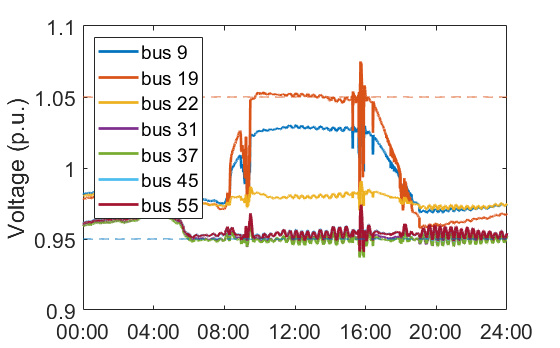}
		\caption{Intermittent communication.}
		\label{fig:robust_asyn}
	\end{subfigure}
	\begin{subfigure}[t]{0.32\textwidth}
		\includegraphics[width=\textwidth]{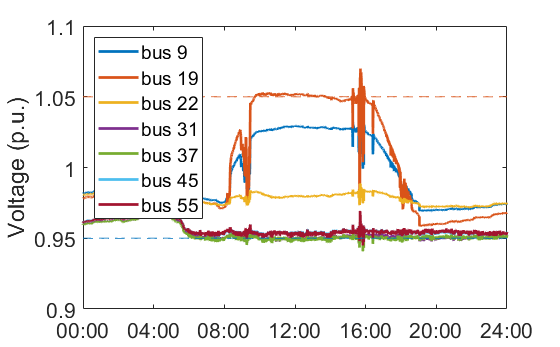}
		\caption{Measurement noise.}
		\label{fig:robust_noise}
	\end{subfigure}
	~
	\begin{subfigure}[t]{0.32\textwidth}
		\includegraphics[width=\textwidth]{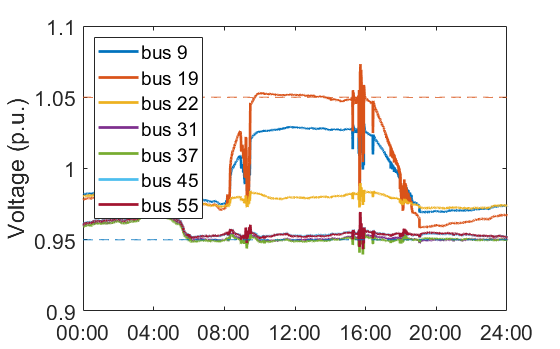}
		\caption{Modelling error.}
		\label{fig:robust_modelerr}
	\end{subfigure}
	\caption{Voltage profile of the  \textbf{DIST-OPT} algorithm under different conditions. }\label{fig:dynamic_v_various_2}
\end{figure*}

We consider the case where a subset of the buses have PV generation (bus 
2, 4, 7, 8, 9, 10, 11, 12, 13, 14, 15, 16, 19, 20, 23, 25, 26, 32). We use the load and PV generation profile  in \cite{bernstein2017real}. The time span of the data set is one day (24 hours), and the time resolution is $6$ seconds. We plot the total load and PV generation profile across the buses in Figure~\ref{fig:load}. 
We assume that there are control components at all the buses and those control components can supply or consume at most $0.1$ p.u. active power and reactive power (i.e. $\bar{p}_i=0.1$, $\underline{p}_i = -0.1$, $\bar{q}_i = 0.1, \underline{q}_i = -0.1$).  The parameter $a_i^P, a_i^Q$ is synthetic data in the range of $[1,2]$, and $b_i^P$ and $b_i^Q$ is set as $0$.   
 Consistent with the time resolution of the dataset, the buses perform one iteration of the  \textbf{DIST-OPT} algorithm every 6 seconds.


 Figure~\ref{fig:dynamic} depicts the progress of the algorithm when there are no communication delays. 
 The total serial running time of our algorithm is 5.383 seconds, running all 14400 iterations on all 55 buses. \footnote{All simulations are done using MATLAB 2018b on a HP Z640 Workstation with Intel Xeon E5-2620 v4 2.10GHz. It should be remarked that the total time of running the simulations was 264.696 seconds. However, 259.313 seconds were spent calculating the power flow through Matpower. This computation is not part of our algorithm. It is simply the output of the physics of the power system and is obtained from measurements when our algorithm is implemented in a real power system.}
 This means that computing 1 iteration takes roughly 7 microseconds on each bus.
 This is negligible  compared to the  6 second time resolution. 
 The resulting active injection and reactive power injection and voltage profile are shown in Figure~\ref{fig:dynamic_p}, \ref{fig:dynamic_q} and \ref{fig:dynamic_v}, respectively. For comparison, we also simulate the voltage profile when no control is applied in Figure~\ref{fig:dynamic_v_withoutcontrol}. 
 A comparison between Figure~\ref{fig:dynamic_v} and Figure~\ref{fig:dynamic_v_withoutcontrol} shows that our algorithm can maintain the voltages within the upper and lower limit over almost the whole day. 
 The only exception is around 16:00 where the voltages overshoot the feasible range for only a short period of time. 
 During this time the PV  generations are changing very  rapidly as can be seen from Figure~\ref{fig:load}.
 However, our algorithm drives the voltages back to the feasible range in only few iterations. 
%
%
 Further, Figure~\ref{fig:dynamic_p} and Figure~\ref{fig:dynamic_q} show that our algorithm does not violate the active and reactive capacity constraints at any time.
We next test the robustness of our algorithm against communication delay, measurement noise and modeling error. In these tests, we use the same simulation setting as that of Figure~\ref{fig:dynamic}.

\textit{Robustness against communication delays.} We test two cases with different types of communication delays. 
In Figure~\ref{fig:robust_delay_fixed}, we set the communication delays between each pair of buses $i$ and $j$ at different times to be a constant $\tau_{ij}(t) = 5$ (30 seconds). In Figure~\ref{fig:robust_delay_random}, the delays $\tau_{ij}(t)$ are drawn independently and uniformly from $[0,\tau_{max}]$, where $\tau_{max}$ is the maximum delay and is set as $15$ (90 seconds). It can be seen that Figure~ \ref{fig:robust_delay_fixed} has no significant difference from Figure~\ref{fig:dynamic_v}, while Figure \ref{fig:robust_delay_random} exhibits small oscillations but are still able to maintain the voltage within the acceptable range. A delay of $5$ iterations means $30$ seconds, and a delay of $15$ iterations means $90$ seconds. These show our algorithm is robust against large communication delays. 

\textit{Robustness against intermittent communication.} In Figure~\ref{fig:robust_asyn}
 we consider intermittent communication, where the nodes communicate only every $5$th iteration (with no communication delay). 
 When no communication occurs then the nodes update their control based on  the last received communicated information.  
 This means that the nodes communicates only once every $30$ seconds. 
 This reduces the communication by $80\%$ compared to communicating at every iteration. 
 %
%
Compared to Figure~\ref{fig:dynamic_v}, Figure~\ref{fig:robust_asyn} exhibits small oscillations but the voltage is still maintained within the acceptable range. These show that when implementing our algorithm, each node does not need to communicate at every iteration, and can simply communicate every a few iterations (e.g. $5$ iterations, 30 seconds as in Figure~\ref{fig:robust_asyn}), which greatly reduces the communication burden. 

\textit{Robustness againt measurement noise.} We test a case where the measurement is corrupted by Gaussian noise with stand error 0.01 p.u. The results are shown in Figure~\ref{fig:robust_noise}. Compared to \ref{fig:dynamic_v}, Figure~\ref{fig:robust_noise} exhibits some small oscillations but still are able to maintain the voltage within the acceptable range. 

\textit{Robustness against modeling error.} We test a case where the $X_{ii}$ and $R_{ii}$ used in the algorithm are inaccurate, and are drawn from $[0.8,1.2]$ of the true value. The results are shown in Figure~\ref{fig:robust_modelerr}. It can be seen that Figure~\ref{fig:robust_modelerr} has no significant difference from Figure~\ref{fig:dynamic_v}.

\begin{remark}
In these simulations the algorithm performs one iteration during every 6 second time window.  
   However, since the computation time is very quick $(\approx 7$ microseconds)  we can easily perform hundreds of iterations per each 6 second time window provided that the communication is fast enough. 
   Even if the communication is slow compared to the computation, these simulations show that it is fine if the nodes do not communicate at every iteration or if they do not wait for the received communication before performing their computation. This means that it is often reasonable to do multiple iterations per time window.  However, the simulations show that performing one iteration per time window is often enough.
\end{remark}

\begin{figure*}
	\centering
	\begin{subfigure}[t]{0.32\textwidth}
		\includegraphics[width=\textwidth]{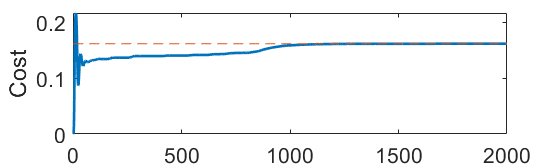}
		\caption{Cost function under $p$-only control.}
		\label{fig:static_ponly_cost}
	\end{subfigure}
	\begin{subfigure}[t]{0.32\textwidth}
		\includegraphics[width=\textwidth]{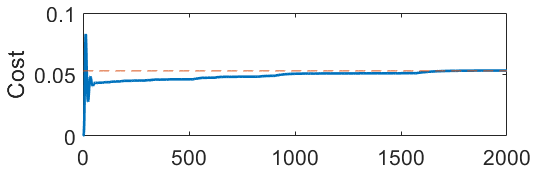}
		\caption{Cost function under $q$-only control.}
		\label{fig:static_qonly_cost}
	\end{subfigure}
	~
	\begin{subfigure}[t]{0.32\textwidth}
		\includegraphics[width=\textwidth]{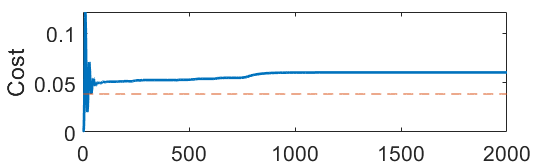}
		\caption{Cost function under joint $p,q$ control.}
		\label{fig:static_pq_cost}
	\end{subfigure}
	\caption{Cost function in the static setting. }\label{fig:static_cost}
\end{figure*}

\subsection{Optimality Under AC Power Model}
 Theorem~\ref{MainTheorem} ensures  that the  \textbf{DIST-OPT} algorithm drives the system operating point to the solution of a optimization problem assuming   that the power flow model is the linear. In this subsection, we test whether the optimality still holds under the nonlinear model, i.e. whether the fixed point of the algorithm is still the solution of the optimization problem~\eqref{MainProblem}  the power flow equation is the full AC model. 


To this end, we run three tests. In all the tests, we use a time-invariant load profile. In the first test, we set for each $i$, $\bar{p}_i = 0.1$, $\underline{p_i} = -0.1$ and $\bar{q}_i = \underline{q}_i = 0$, i.e., we only use active power $p$ injection to do the control. In the second test, we set $\bar{p}_i = \underline{p}_i = 0$ $\bar{q}_i = 0.1$, $\underline{q_i} = -0.1$, i.e., we use only reactive power $q$ to do the control. In the third test, we set $\bar{p}_i = 0.1$, $\underline{p_i} = -0.1$, $\bar{q}_i = 0.1$, $\underline{q_i} = -0.1$, i.e., we use both active power and reactive power to do the control. 
The cost function of the three cases are shown in Figure~\ref{fig:static_cost}, where the dashed line depict the optimal solution of \eqref{MainProblem} under nonlinear AC Power Flow Model~\eqref{eq:LinBranchFlow}, using the SOCP relaxation in \cite{low2014convex}. It is  seen from Figure~\ref{fig:static_ponly_cost} and Figure~\ref{fig:static_qonly_cost} that if we do only active power control or only reactive power control then \textbf{DIST-OPT}  drives the system to the optimum whereas if we do joint active-reactive power control then \textbf{DIST-OPT} drives the system to a non-optimum point. To further support the above observation we re-do the three tests under $10$ randomly generated load conditions. We consider the relative absolute error 
$$ \Big|\frac{Cost - Opt}{Opt}\Big| \times 100\% $$
where $Cost$ is the cost function achieved by running \textbf{DIST-OPT} for 4000 iterations, and $Opt$ is the optimal solution obtained by the SOCP relaxation. 
Our results indicate that the active-power-only and reactive-power-only controls achieve an 2.7\% and 4.3\% relative absolute error on average, respectively. 
%
%
  However, joint active-reactive control achieves an 59.92\% relative absolute error on average.

All these tests show that when doing active-power-only or reactive-power-only control, \textbf{DIST-OPT} can drive the system to (nearly) the optimum, and hence the conclusion of Theorem~\ref{MainTheorem}  still holds under nonlinear AC power flow models. However, when doing joint active-reactive control, \textbf{DIST-OPT} may fail to reach the optimum. 
%
We conjecture this may be due to that the linearized power flow model~\eqref{eq:PhysicalRelationship-1} that we used to develop our algorithm does not capture  well some nonlinearities in the coupling between $p$ and $q$ in the full AC mode~\eqref{eq:PhysicalRelationship-1}. 
 Nevertheless, our algorithm does a good voltage regulations when we do a joint $p$ and $q$ control, as illustrated by our experiments in the previous section. 
  Finally, we comment that reaching the optimal solution of an optimal power flow problem with nonlinear AC power flow equation is a very difficult non-convex problem, and to date there has been only limited theoretic understanding \cite{low2014convex}. Our results only empirically show \textbf{DIST-OPT} may reach the optimal solution under some circumstances. However to theoretically understand the optimality of  \textbf{DIST-OPT} under nonlinear AC power flow remains challenging and interesting future work.

\subsection{Test on the  extended  algorithm}
In this subsection, we test the more generalized algorithm discussed in Section~\ref{subsec:general}. We use the same setting as that in Figure~\ref{fig:dynamic} in Section~\ref{subsec:simu_timevarying}, except that we add the following apparent power constraint,
$$p_i^2 + q_i^2\leq \bar{s}_{i}^2 $$ 
where $\bar{s}_{i}$ is set as $0.12$ (slightly larger than the box constraint on $p_i, q_i$). We run the extended algorithm described in Section~\ref{subsec:general} with the modified STEP 2 in \eqref{eq:general_constr}. The resulting voltage trajectory is given in Figure~\ref{fig:dynamic_ap_v}, and trajectories of active power $p_i$, reactive power $q_i$, and apparent power $s_i = \sqrt{p_i^2 + q_i^2}$ are given in Figure~\ref{fig:dynamic_ap_pqs}. The results show that the extended algorithm can meet the additional apparent power constraint while still guaranteeing the voltage lies between the upper and lower limit.

\begin{figure}
	\includegraphics[width=\columnwidth]{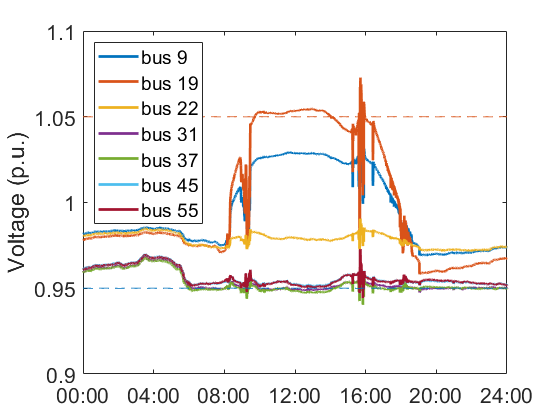}
	\caption{Voltage profile of the extended \textbf{DIST-OPT} algorithm with apparent power constraint.}\label{fig:dynamic_ap_v}
\end{figure}

\begin{figure*}
	\centering
	\begin{subfigure}[t]{0.32\textwidth}
		\includegraphics[width=\textwidth]{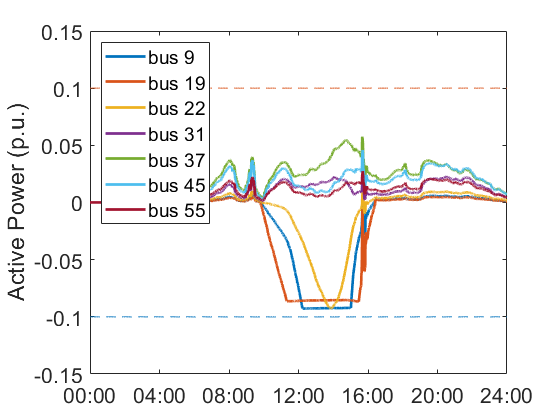}
		\caption{Active power.}
	\end{subfigure}
	\begin{subfigure}[t]{0.32\textwidth}
		\includegraphics[width=\textwidth]{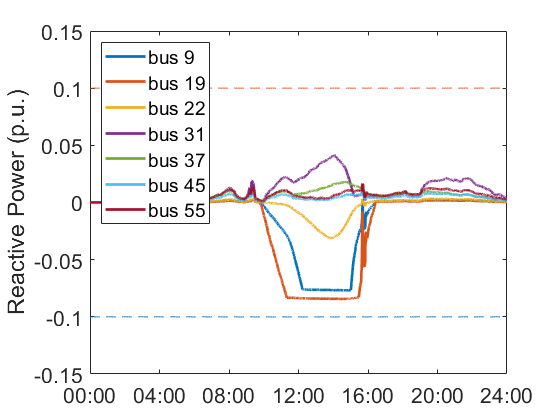}
		\caption{Reactive power.}
	\end{subfigure}
	~
	\begin{subfigure}[t]{0.32\textwidth}
		\includegraphics[width=\textwidth]{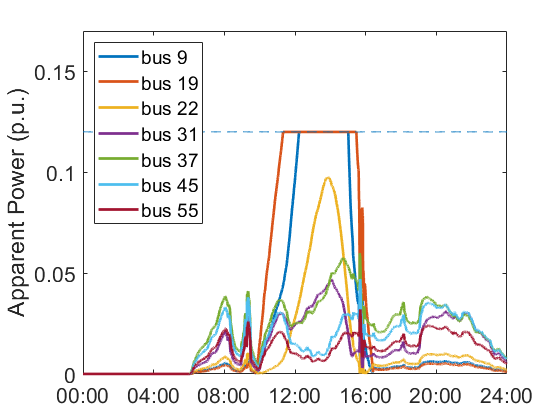}
		\caption{Apparent power.}
	\end{subfigure}
	\caption{Active, reactive, apparent power profile of the extended \textbf{DIST-OPT} algorithm with apparent power constraint. }\label{fig:dynamic_ap_pqs}
\end{figure*}



\section{Colusion}

 We studied distributed voltage control algorithms where the buses  perform local voltage control based only on local measurements and  communication to neighbours in the network. 
  We proved that our algorithms converge to an optimal voltage profile under linear power flow model even if the communication is asynchronous or delayed. 
  The good performance of our algorithm and  its robustness to asynchronous communications was further illustrated in simulations under realistic operation conditions using the full nonlinear AC power flow model. 
  Our simulations showed that our algorithm could reduced $80\%$ of the communication compared to a synchronous algorithm while achieving similar  performance. 
  In future work we will study how we can extend our algorithms to cover more dynamic loads such as vehicle charging or smart appliances.


\bibliographystyle{IEEEbib}
\bibliography{refs}

\appendices

\section{Convergence: Proof Theorem~\ref{MainTheorem}} \label{Sec:Conv}

 We now show that the algorithm converges to the optimal solution to Problem~\eqref{MainProblem}.
 The proof follows similar ideas as used in~\cite{Low_1999} to prove the convergence of asynchronous dual decomposition for internet data flow.
 However, we note that our problem is not a special case of the internet data flow problem and the proof ideas need to be adjusted to our problem to work. 
 In particular, the proof is based on the following two lemmas, proved in appendices~\ref{APPLemma:Gradient} and~\ref{APP:Lemma:Descent}, respectively.
\begin{lemma}  \label{Lemma:Gradient}
 For all $t\in \N$ following holds:
 \begin{align*}
    ||\nabla D(\bla(t)){-}g(t)|| \leq L\sqrt{N} \sum_{\tau=t-t_0}^{t-1} ||\bla(\tau){-}\bla(\tau+1)||,
 \end{align*} 
 where $t_0=d(\tau_{\max}+1)$,  $$L=2\left(\frac{||R||^2+||X||^2}{a_{\min}}\right),$$ and $g(t)$ is the approximate dual gradient in Equation~\eqref{eq:dualGradApprox}.
\end{lemma} 
\begin{lemma}  \label{Lemma:Descent}
 For all $t\in \N$ following holds:
  \begin{align*}
     D(\bla(t{+}1)) \geq& D(\bla(0)) {+} \bigg(\frac{1}{\gamma}{-}\frac{L}{2} {-}(\tau_{\max}(d{+}1){+}1) L \sqrt{N} \bigg)  \\
                                  &~~~~~~~~~~~~~~~~~~~ \times~\sum_{\tau=0}^{t} ||\bla(\tau+1)-\bla(\tau)||^2  
  \end{align*}
  In particular, from the Slaters condition in Equation~\eqref{Eq:Slater} the duality gap is zero and, hence, if  $\gamma$ is chosen as in Equation~\eqref{eq:MainStepSize} then
  $$\sum_{\tau=0}^{\infty} ||\bla(\tau{+}1)-\bla(\tau)||^2<\infty  \text{ and }\lim_{t\rightarrow\infty}  ||\bla(t{+}1)-\bla(t)||=0.$$
\end{lemma} 

 The two lemmas show that the  approximate $g(t)$ converges to the true gradient $\nabla D(\bla(t))$ as $t$ goes to infinity, i.e.,
 $$ \lim_{t\rightarrow \infty }||\nabla D(\bla(t))-g(t)||=0.$$
 In particular, Lemma~\ref{Lemma:Gradient} shows that the distance between  $g(t)$ and $\nabla D(\bla(t))$  is bounded by the finite sum
 $$\sum_{\tau=t-t_0}^{t-1} ||\bla(\tau)-\bla(\tau+1)||,$$
 times a constant factor. 
Lemma~\ref{Lemma:Descent}  the shows that the terms of the sum converge to zero when $\gamma$ is chosen as in Equation~\eqref{eq:MainStepSize}. 
 Therefore, the sum also converges to $0$, since it has only $t_0$ terms. 
 We now use these results to proof the theorem.
 
 \underline{\textbf{Proof of Theorem~\ref{MainTheorem}:}}
   We start by showing that every limit point of $\bla(t)$ is an optimal solution to the dual problem.
  Note that the sequence $\bla(t)$ is bounded since from Lemma~\ref{Lemma:Descent}
   $$ \bla(t) \in\{\bla \in \R_+^{2N} | D(\bla)\geq D(\bla(0))\}~~\text{ for all }~~t\in \N$$
   and every level sets is bounded~\cite[Proposition B.9]{nonlinear_bertsekas}.\footnote{Note that from Slaters condition (Equation~\eqref{Eq:Slater}) the set of optimal solutions to the dual problem is bounded, see Lemma~1 in~\cite{nedic2009approximate}.} 
  Let $\bla^{\star}$ be some limit point of $\bla(t)$ and let $\bla(t_i)$ be a subsequence that converges to $\bla^{\star}$.  
  Then we have from the continuity of $\nabla D(\cdot)$ that 
  $ \lim_{i \rightarrow \infty } \nabla D(\bla(t_i))=\nabla D(\bla^{\star}).$ 
  The gradient  approximate sequence $g(t_i)$ also converges to $\nabla D(\bla^{\star})$, since from  the triangle inequality we get
  \begin{multline*}
     \lim_{i \rightarrow \infty} ||\nabla D(\bla^{\star})-g(t_i))|| \leq \lim_{i \rightarrow  \infty}  ||\nabla D(\bla^{\star})-\nabla D(\bla(t_i))||  \\ 
       + \lim_{i \rightarrow \infty }||\nabla D(\bla(t_i))-g(t_i)||  =0
  \end{multline*}
   where the second limit convergence to zero because of lemmas~\ref{Lemma:Gradient} and~\ref{Lemma:Descent}. 
   Therefore, we have 
   \begin{align*} 
      \lceil \bla^{\star}+\gamma \nabla D(\bla^{\star})\rceil_+ -\bla^{\star} =& \lim_{i\rightarrow \infty}   \lceil \bla(t_i)+\gamma g(t_i)\rceil_+-\bla(t_i) \\
      =&  \lim_{i\rightarrow \infty}  \bla(t_i+1)-\bla(t_i)=0,
   \end{align*}
   from Lemma~\ref{Lemma:Descent}. From the projection theorem~\cite[Proposition 2.1.3]{nonlinear_bertsekas} we have that
   $$ \langle \nabla D(\bla^{\star}),\bla^{\star}-\bla)\geq 0,~~\text{ for all }~~ \bla \in \R^{2N},$$
   which implies that $\bla^{\star}$ is the optimal solution to the dual problem~\cite[Proposition 2.1.2]{nonlinear_bertsekas}. 
   
 We can now show that $q(t)$ converges to $q^{\star}$.  
 The sequence $q(t)$ is bounded since it is in $[\un{q} ,\bar{q}]$. 
 Moreover, since the function $q(\cdot)$ is continuous, see Equation~\eqref{eq:primal_update_Cent}, and from strong duality, every subsequence of $q(t)$ convergences to $q^{\star}=q(\bla^{\star})$. 
 Therefore, we can conclude that $q(t)$ converges to $q^{\star}$.

\section{Proof of Lemma~\ref{Lemma:DualMain}} \label{APP:Lemma:DualMain}

 The gradient of the Lagrangian function in Equation~\eqref{EQ:LAG} with respect to $p$ and $q$ is
 \begin{align*}
   \nabla\mathcal{L}(p,q,\bla)
   {=} \left[\begin{array}{c}\nabla^p \mathcal{L}(p,q,\bla) \\ \nabla^q \mathcal{L}(p,q,\bla)\end{array}\right]
   {=}
   \left[\begin{array}{c}
      \BLambda\re p{+}b\re{+}R(\bar{\Blambda}{-}\un{\Blambda}) \\
       \BLambda\im q{+}b\im{+}X(\bar{\Blambda}{-}\un{\Blambda})
  \end{array}\right].
 \end{align*}
 From Proposition~6.1.1 in~\cite{nonlinear_bertsekas} the dual gradient is  
\begin{align} 
  \nabla  D(\Blambda) = \left[\begin{array}{c}   \vu-\v(p^{\star}(\bla),q^{\star}(\bla)) \\ \v(p^{\star}(\bla),q^{\star}(\bla) )-\vb  \end{array}\right] 
\end{align} 
  where \vspace{-0.2cm}
 \begin{equation}\label{eq:inProofOpt}
     (p^{\star}(\bla), q^{\star}(\bla) )= \underset{(p,q)\in[\pu,\pb]\times [\qu,\qb]}{\text{argmin }} \mathcal{L}(p,q,\bla). \vspace{-0.2cm}
  \end{equation}
  Therefore, to prove Equation~\eqref{eq:dualGrad} it suffices to   show that 
  $$p^{\star}(\bla)=p(\bla) ~~\text{ and }~~\q^{\star}(\bla)=\q(\bla)$$
  or,  equivalently, to show that $(p(\bla),q(\bla))$ is an optimal solution to Problem~\eqref{eq:inProofOpt}  (its solution is unique since $a_{\min}>0$). 
   By  Proposition 2.1.2. in~\cite{nonlinear_bertsekas} $(p(\bla),q(\bla))$ is an optimal solution to Problem~\eqref{eq:inProofOpt}   if and only if 
\begin{align*} 
      \left\langle  \nabla\mathcal{L}(p(\bla),q(\bla),\bla),
        \left[\begin{array}{c}   p{-}p(\bla) \\
                                                         q{-}q(\bla)
        \end{array} \right]     
         \right\rangle\geq 0 
\end{align*}
  for all $(p,q)\in [\pu,\pb] \times  [\qu,\qb]$. We have that
   \begin{align*}
            \left\langle  \nabla\mathcal{L}(p,q,\bla),
        \left[\begin{array}{c}   p{-}p(\bla) \\
                                                         q{-}q(\bla)
        \end{array} \right]     
         \right\rangle {=}&   \sum_{i=1}^N \nabla_i^p \mathcal{L}(p,q,\lambda) (p_i-p_i(\bla)) \\
                               &{+} \sum_{i=1}^N \nabla_i^q \mathcal{L}(p,q,\lambda) (q_i{-}q_i(\bla))
   \end{align*}
   and hence  it suffices to prove that each term of the two sums above is positive.  
   We prove that each term of the first sum is positive,  the proof for the second sum is identical. 
   In particular, we prove that each term of the  sum is positive by considering separately the following three cases:
   \begin{enumerate}[\hspace{2pt}a)]
     \item If  $p_i(\bla)\in (\un{p}_i,\bar{p}_i)$  then we show that 
     $$\nabla_i^p \mathcal{L}(p(\bla),q(\bla),\q(\bla),\bla)=0.$$ 
     \item 
     If $p_i(\bla)=\un{p}_i$ then we show that
     $$\nabla_i^p \mathcal{L}(p(\bla),q(\bla),\bla)\geq 0.$$
     \item  If $p_i(\bla)=\bar{p}_i$ then we show that 
     $$\nabla_i^p \mathcal{L}(p(\bla),q(\bla),\bla)\leq 0.$$ 
    \end{enumerate}
   We now conclude the proof by proving a), b), and c) below.  
   
  \underline{Prove of  a):} 
    Note that $p(\bla)=\left[p_{\text{UC}}^{\star}(\bla) \right]_{\pu}^{\pb}$ where
 $$p_{\text{UC}}^{\star}(\bla):= \underset{p\in \R^n}{\text{argmin }} \mathcal{L} (p,q,\bla) =  {\Lambda_{\texttt{P}}}^{-1} R( \BLLambda -\BULambda) -b\re     $$
 is the unconstrained optimizer of $\mathcal{L}(\cdot,q,\bla)$ (the optimal solution is independent of $q$).
  If $p_i(\bla)\in (\un{p}_i,\bar{p}_i)$ then $p_i(\bla)=[p_{\text{UC}}^{\star}(\bla)]_i$.  
  Using that $\BLambda_{\texttt{P}}$ is a diagonal we have   
  \begin{align*}
      \nabla_i^p \mathcal{L}(p(\bla),q(\bla),\bla)  =&  a_i\re p_i(\bla)+b_i\re+[R(\bar{\Blambda}-\un{\Blambda})]_i  \\
                                                                    =&  a_i\re [p_{\text{UC}}^{\star}(\bla)]_i+b_i\re+[R(\bar{\Blambda}-\un{\Blambda})]_i = 0                             
  \end{align*}
  since $\nabla_i^p \mathcal{L}(p^{\star}_{\text{UC}}(\bla),q,\bla)   =0$ and $p_{\text{UC}}^{\star}(\Blambda)$ is the  optimizer.  
  
  \underline{Prove of  b) and c):}  If $p_i(\bla)=\un{p}_i$ then $[p_{\text{UC}}^{\star}(\bla)]_i \leq p(\bla)$. Therefore, since $a_i\re>0$ we have
    \begin{align*}
     0 =& \nabla_i^p \mathcal{L}(p_{\text{UC}}^{\star}(\bla),q,\bla)  
        =  a_i\re [p_{\text{UC}}^{\star}(\bla)]_i+b_i\re+[R(\bar{\Blambda}-\un{\Blambda})]_i, \\
        \leq&  a_i\re p_i(\bla)+b_i\re+[R(\bar{\Blambda}-\un{\Blambda})]_i  = \nabla_i^p \mathcal{L}(p(\bla),q,\bla)
  \end{align*}
   Condition c) follows from similar arguments as condition b).

 Finally we  show that $\nabla D(\bla)$ is  $L$-Lipschitz continuous. 
 Take $\bla_1=(\un{\Blambda}_1,\bar{\Blambda}_1), \bla_2=(\un{\Blambda}_2,\bar{\Blambda}_2)\in \R_+^{2n}$, then from Equations~\eqref{eq:dualGrad} and~\eqref{eq:LocalProblem-c}  
 \begin{align*}
    ||\nabla D(\bla_1){-}\nabla D(\bla_2)|| \leq& \sqrt{2}|| \v(p(\bla_1),\q(\bla_1)){-}\v(p(\bla_1),\q(\bla_2))|| \\
                     \leq& \sqrt{2} \big(||R||~|| p(\bla_1)-p(\bla_2)|| \\ &~~~~~~~~~~~ +||X||~|| \q(\bla_1)-\q(\bla_2)||)\\
                     \leq& \sqrt{2} \left( \frac{||R||^2 {+}||X||^2}{a_{\min}} \right)~|| \un{\lambda}_1{-}\un{\lambda}_2{+}\bar{\lambda}_2{-}\bar{\lambda}_1||\\
                    \leq&   2 \left( \frac{||R||^2 +||X||^2}{a_{\min}} \right) ~|| \bla_1-\bla_2||,
 \end{align*} 
  where we have used the triangle inequality in  the first and last inequality and the fact that $||\Lambda_{\texttt{P}}^{-1}||,||\Lambda_{\texttt{Q}}^{-1}||\leq1/a_{\min}$ to obtain the 3rd inequality.

\section{Proof of Lemma~\ref{lemma:delayEqu}} \label{APP:lemma:delayEqu}

 We need the following definition.
\begin{defin}
   Consider a rooted tree $\mathcal{T}=(\mathcal{N},\mathcal{E})$. 
\begin{itemize}
  \item   For each node $i\in \mathcal{N}$ we define the set of \textbf{$r$-th descendants} of $i$ as follows 
 $$\mathcal{C}_i^r=\{ j\in \mathcal{N} : \sigma^r(j)=i\}.$$
  Moreover, define the set of $i$ and all of its descendants  as follows \vspace{-0.3cm}
 \begin{align*}
      \mathcal{D}(i)=&\bigcup_{k=0}^{\infty}  \mathcal{C}_i^r 
  \end{align*}
 \item We define the \textbf{height} of  a node $i\in \mathcal{N}$ as follows 
  $$\texttt{Height}(i)=\max\{n\in \N : C_i^n \neq \emptyset  \}.$$   

\item  We define the \textbf{depth} of node $i\in \mathcal{N}$ 
as the distance from $i$ to the root node $R\in \mathcal{R}$, i.e., 
  $$\texttt{Depth}(i)=\dist(i,R).$$

\item  We define the \textbf{most recent common ancestor} of nodes $i,j\in \mathcal{N}$ as follows 
  $$\texttt{MRCA}(i,j)=\underset{k\in \mathcal{A}_i \cap \mathcal{A}_j}{\text{argmax}} ~\texttt{Depth}(k),$$
  where $\mathcal{A}_i=\{k\in \mathcal{N}: \sigma^r(i)=k \text{ for some }r\in \N\}$ is the set of ancestors of node $i$.  
\end{itemize}
\end{defin}
 Using the notation from the definition, we  have the following claims (proved in the sequel):
\begin{itemize}
  \item \textbf{Claim 1:} For $\alpha_i(t)$ defined in Equation~\eqref{EQ:AlphaUpdate}   we have \vspace{-0.1cm}
   \begin{align*}  
      \lambda_i(t)+\sum_{j\in \mathcal{C}_i} \alpha_j(t) = \sum_{j\in \mathcal{D}(i)} \lambda_{j}(t-d_{ij}),
   \end{align*}
 where we set $\lambda_j(t)=0$ for $t<0$.
  \item \textbf{Claim 2:} For $\beta_i(t)$ defined in Equation~\eqref{EQ:BetaUpdate}   we have 
  $$\beta_i(t)= \sum_{k=1}^{\texttt{Depth}(i)} \chi_{\sigma^k(i)} \sum_{j\in \mathcal{D}(i,k)} \lambda(t-d_{ij}),$$
  where $\chi=[X_{11},\ldots,X_{NN}]$ and $\mathcal{D}(i,k)=\mathcal{D}(\sigma^k(i)) \setminus \mathcal{D}(\sigma^{k-1}(i))$ and $\lambda_j(t)=0$ for $t<0$.
 \item \textbf{Claim 3:} We have 
     $X_{ij}=X_{kk}$ where $k=\texttt{MRCA}(i,j)$, 
    i.e., $X_{ij}=\chi_{\texttt{MRCA}(i,j)}$.
\end{itemize}
 Plug in the equations from the three claims into Equation~\eqref{eq:MainZ} proves equations~\eqref{Eq:Delay_Z} and~\eqref{Eq:Delay_Z2}.
  We now prove the claims.

\underline{Prove of \textbf{Claim 1}:}
  The equation follows from the following equation (proved in the sequel) \vspace{-0.1cm}
   \begin{align*}  
       \alpha_i(t)=&\lambda_i(t)+\sum_{r=1}^{\texttt{Height}(i)} \sum_{j\in \mathcal{C}_i^r} \lambda_j(t-r), \\
                      =& \sum_{j\in \mathcal{D}(i)} \lambda_{j}(t-\dist(i,j)).
   \end{align*}
    We proof the result by induction over $\texttt{Height(i)}$. Suppose first that $\texttt{Height(i)}=0$, i.e., node $i$ is a leave. 
  Then the result follows from Equation~\eqref{EQ:AlphaUpdate}. Suppose now that the equation holds for $\texttt{Height(i)}=r$.  
 Then from Equation~\eqref{EQ:AlphaUpdate} 
 \begin{align*}
    \alpha_i(t)  =& \lambda_i(t)+\sum_{j\in \mathcal{C}_i} \alpha_j(t)   \\
                     =&  \lambda_i(t)+\sum_{j\in \mathcal{C}_i} \lambda_j(t-1) + \sum_{r=2}^{\texttt{Height}(i)} \sum_{j\in \mathcal{C}_i^r} \lambda_j(t-r)  \\
                    =& \lambda_i(t)+\sum_{r=1}^{\texttt{Height}(i)} \sum_{j\in \mathcal{C}_i^r} \lambda_j(t-r) 
 \end{align*}
   where we use the induction premises in the second equality.

\underline{Prove of \textbf{Claim 2}:} Writing out the recursion in Equation~\eqref{EQ:BetaUpdate} and using that $\beta_j(t)=0$, for all $t$, if $i$ if $i$ has no parent (i.e., if $j$ is the root) then we get
 \begin{align*}
    \beta_i(t){=}&  \hspace{-0.2cm}  \sum_{k=1}^{\texttt{Depth}(i)}  \hspace{-0.2cm}  \chi_{\sigma^k(i)} \left(\lambda_{\sigma^k(i)}(t+1-k)+ \hspace{-0.7cm} \sum_{ r\in \mathcal{C}_{\sigma^k(i)} \setminus \{\sigma^{k-1}(i)\}}  \hspace{-0.7cm}  \alpha_j(t-k) \right) \\ 
     {=}&  \hspace{-0.2cm}  \sum_{k=1}^{\texttt{Depth}(i)}  \vspace{-0.2cm}  \chi_{\sigma^k(i)} \vspace{-0.4cm} \sum_{j\in \mathcal{D}(i,k) }  \vspace{-0.5cm}  \lambda_{\sigma^k(i)}(t{+}1{-}(k{+}\overbrace{\dist(\sigma^k(i),j)}^{=\dist(i,j)}))
 \end{align*}

\underline{Prove of \textbf{Claim 3}:} Follows from that $\X_{ij}=2 \sum_{(h,k)\in\mathcal{P}_i\cap \mathcal{P}_j} x_{hk}$ and that $\texttt{MRCA}(i,j)$ is the end point of the intersection of the two paths $\mathcal{P}_i$  and $\mathcal{P}_j$.

 \section{Proof of Lemma~\ref{Lemma:Gradient}}  \label{APPLemma:Gradient}

  From equations~\eqref{eq:dualGrad} and~\eqref{eq:dualGradApprox}, we have
  \begin{align*}
    ||\nabla D(\bla(t)){-}g(t)|| =& \sqrt{2}  || v(q(\bla(t)))-v(q(t))|| \\
                                           \leq& \sqrt{N} L \sum_{\tau=t-t_0}^{t-1} || \bla(\tau){-}\bla(\tau{+}1)||  
  \end{align*}
  where the $\sqrt{2}$ factor in the first equation comes from the duplication of $v(\cdot)$ in equations~\eqref{eq:dualGrad} and~\eqref{eq:dualGradApprox} and the second equation comes from the following three inequalities (proved below): 
 \begin{align}
    &||v(q(\bla(t))){-}v(q(t))||  \leq  ||R||   || p(\bla(t))-p(t)|| \notag \\ &\hspace{4.2cm} + ||X||   || q(\bla(t))-q(t)||,  \label{eq:VdelayBound}\\
    & \hspace{-0.25cm} ||p(\bla(t)){-}p(t)|| {\leq}  \frac{ \sqrt{2N} ||R||}{a_{min}}   \sum_{\tau=t-t_0}^{t-1} || \bla(\tau){-}\bla(\tau{+}1)|| , \label{eq:PdelayBound}    \\
    & \hspace{-0.25cm} ||q(\bla(t)){-}q(t)|| {\leq}  \frac{\sqrt{2N} ||X||}{a_{min}}   \sum_{\tau=t-t_0}^{t-1} || \bla(\tau){-}\bla(\tau{+}1)|| . \label{eq:QdelayBound} 
 \end{align}

\noindent  Equation~\eqref{eq:VdelayBound} follows from the definition of $v(\cdot)$ in Equation~\eqref{Eq:linearizedV} and the fact that $||\cdot||$ is the induced matrix norm. 
 To prove Equation~\eqref{eq:PdelayBound}, Equation~\eqref{eq:QdelayBound} is provide similarly,
  we recall that from Equation~\eqref{eq:LocalProblem-c} and Lemma~\ref{lemma:delayEqu} we have  \vspace{-0.1cm}
  $$p_i(t)=\left[  \frac{1}{a_i\re} \sum_{j=1}^N R_{ij} \lambda_j(t-\bar\tau_{ji}(t)) -b_i\re  \right]_{\un{q}_i}^{\bar{q}_i},$$
  where $\bar\tau_{ij}(t)\leq t_0=d(\tau_{\max}+1)$. 
  Therefore, focussing on component $i$ of the vector $p(\bla(t))-p(t)$ and using the non-expansion property of the projection we get  \vspace{-0.2cm}
  \begin{align*}  
       |p_i(\bla(t))-p_i(t)| 
               \leq&  \frac{||R||}{a_{\min}}   \sum_{j=1}^N |\lambda_j(t)-\lambda_j(t-\bar\tau_{ji}(t)) |\\
               \leq& \frac{||R||}{a_{\min}}  \sum_{j=1}^N \sum_{k=t-\bar\tau_{ji}}^{t-1}   |\lambda_j(k+1)-\lambda_j(k) | \\
                    \leq& \frac{||R||\sqrt{N}}{a_{\min}}  \sum_{k=t-t_0}^{t-1}  || \lambda(k+1)-\lambda(k) ||           \\
                    \leq& \frac{ ||R||\sqrt{2 N}}{a_{\min}}  \sum_{k=t-t_0}^{t-1}  || \bla(k+1)-\bla(k) ||           
  \end{align*}  
  where 
  the first inequality comes by the definitions of $p_i(\bla(t))$ and $p_i(t)$, the the triangle inequality, and the fact that $1/a_i\re \leq 1/a_{\min}$ for all $i$.
  The second inequality comes by using the triangle inequality. 
  The  third inequality comes by adding extra terms to the inner sum (every term is positive) so it runs from $k=t-t_0$ to $t$, swapping the sums, and using the equivalence of norms, i.e., $||\cdot||_1\leq \sqrt{N}||\cdot ||$. 
  The final inequality is obtained by noting that $\lambda(k)=\un{\lambda}(k){-}\bar{\lambda}(k)$ so
  \begin{align*}
   \hspace{-0.3cm}  || \lambda(k{+}1){-}\lambda(k) ||^2 
          \hspace{-0.1cm} \leq & 2( || \un{\lambda}(k{+}1){-}\un{\lambda}(k)||^2{+}|| \bar{\lambda}(k{+}1){-}\bar{\lambda}(k)||^2 ) \\
           =& 2 ||\bla(k+1)-\bla(k)||^2. 
  \end{align*}
  Equation~\eqref{eq:QdelayBound} can now be obtained by using the equivalence of the $||\cdot||_{\infty}$ and $||\cdot||$ norms as in the prove of  Equation~\eqref{eq:VdelayBound}.

\section{Proof of Lemma~\ref{Lemma:Descent}} \label{APP:Lemma:Descent}

Set $\Delta(k)=\bla(k+1)-\bla(k)$.
 From the convexity of $-D(\cdot)$ we have~\cite[Theorem 2.1.5]{Book_Nesterov_2004} \vspace{-0.2cm}
 \begin{align*}
     -D(\bla(t{+}1)) \leq&{-} D(\bla(t)){-}\langle \nabla D(\bla(t)), \Delta(t)\rangle {+} \frac{L}{2} || \Delta(t) ||^2  \\
                           =& - D(\bla(t))-\langle \nabla D(\bla(t))-g(t),\Delta(t)\rangle  \\
                                 &-\langle g(t), \Delta(t) \rangle+ \frac{L}{2} || \Delta(t) ||^2\\
                           \leq&    - D(\bla(t))+|| \nabla D(\bla(t))-g(t)||~ ||\Delta(t)|| \\
                                 &  -\left(\frac{1}{\gamma}-\frac{L}{2}\right) ||\Delta(t)||^2, 
 \end{align*}
 where in the last inequality we have used
     $ \frac{1}{\gamma} ||\Delta(t)||^2\leq \langle g(t),\Delta(t)  \rangle$, 
  which is obtained by noting that $\bla(t+1)=  \lceil \bla(t)+ \gamma   g(t) \rceil_+$ (Equation~\eqref{eq:delayed dual descent}) 
    and hence  from the projection theorem in~\cite[Lemma 2.1.3 (b)]{nonlinear_bertsekas} we have  \vspace{-0.1cm}
  \begin{align*} 0\geq&  \langle  \bla(t)+ \gamma   g(t)-  \bla(t+1), \bla(t) -  \bla(t+1) \rangle \\
                        =& -\gamma \langle     g(t),\Delta(t) \rangle + ||\Delta(t)||^2.
  \end{align*}
  Expanding further by using Lemma~\ref{Lemma:Gradient} we get \vspace{-0.1cm}
 \begin{align*}
    -D(\bla(t{+}1)) \leq& - D(\bla(t))  -\left(\frac{1}{\gamma}-\frac{L}{2}\right) ||\Delta(t)||^2  \\ 
                                 & +L\sqrt{N} \sum_{k=t-t_0}^{t-1} ||\Delta(k)||~  ||\Delta(t)|| \\
                           \leq& - D(\bla(t)) -\left(\frac{1}{\gamma}-\frac{L}{2}\right) ||\Delta(t)||^2  \\
                                 & + L\sqrt{N}  \sum_{\tau=t-t_0}^{t} ||\Delta(k)||^2,
 \end{align*}
 where 
  the final inequality is obtained by using the fact that for any $s_{t-t_0},\ldots,s_t\in \R_+$ it holds that 
  $ \sum_{k=t-t_0}^{t-1} s_k s_t  \leq \sum_{k=t-t_0}^{t} s_k^2 .$
 If we sum over $t$ we get  \vspace{-0.3cm}
  \begin{align*}
     -D(\bla(t{+}1))  \leq& - D(\bla(0)) -\left(\frac{1}{\gamma}-\frac{L}{2}\right) \sum_{k=0}^{t} ||\Delta(k)||^2  \\
                                   & +L\sqrt{N} \sum_{k_1=0}^{t}\sum_{k_2=\tau_1-t_0}^{k_1}   ||\Delta(k)||^2, \\                     
                          \leq& {-} D(\bla(0)) {-} \left(\frac{1}{\gamma}-\frac{L}{2}-(t_0{+}1) L\sqrt{N}  \right)   \allowdisplaybreaks \\
                                  &~~~\times \sum_{\tau=0}^{t} ||\bla(\tau+1)-\bla(\tau)||^2,  
 \end{align*}
 which concludes the proof.

\begin{IEEEbiography}[{\includegraphics[width=1in,height=1.25in,clip,keepaspectratio]{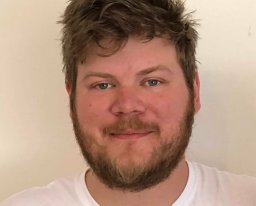}}]{Sindri Magn\'usson}
 received the B.Sc. degree in Mathematics from University of Iceland, Reykjavík Iceland, in 2011, the Masters degree in Applied Mathematics (Optimization and Systems Theory) from KTH Royal Institute of Technology, Stockholm Sweden, in 2013, and the PhD in Electrical Engineering from the same institution, in 2017. He was a postdoctoral researcher 2018-2019 at Harvard University, Cambridge, MA and a visiting PhD student at Harvard University for 9 months in 2015 and 2016. His research interests include large scale distributed/parallel optimization, machine learning, and control, both theory and applications.
\end{IEEEbiography}

 \begin{IEEEbiography}[{\includegraphics[width=1in,height=1.25in,clip,keepaspectratio]{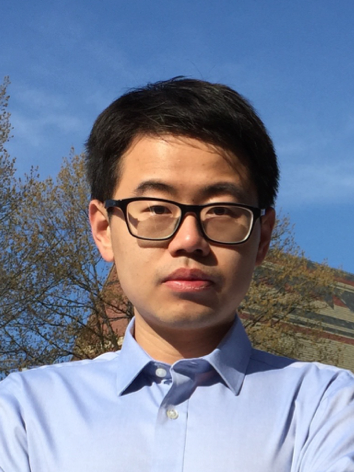}}]{Guannan Qu} 
  received his B.S. degree in Electrical Engineering from Tsinghua University in Beijing, China in 2014, and his Ph.D. from Harvard University in 2019. Since 2019 he has been a postdoctoral scholar in the Department of Computing and Mathematical Sciences at California Institute of Technology. His research interest lies in control, optimization, and learning particularly in network systems.
 \end{IEEEbiography}

 \begin{IEEEbiography}[{\includegraphics[width=1in,height=1.25in,clip,keepaspectratio]{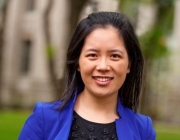}}]{Na Li} 
received the B.S. degree in mathematics and applied mathematics from Zhejiang University, Hangzhou, China, in 2007 and the Ph.D. degree in control and dynamical systems from the California Institute of Technology, Pasadena, CA, USA, in 2013. She is an Associate Professor with the School of Engineering and Applied Sciences, Harvard University. She was a Postdoctoral Associate with the Laboratory for Information and Decision Systems, Massachusetts Institute of Technology. Her research interests include the design, analysis, optimization, and control of distributed network systems, with particular applications to cyber-physical network systems. She received National Science Foundation CAREER Award in 2016, AFOSR Young Investigator Award (2017), Office of Naval Research Young Investigator Award in 2019, Donald P. Eckman Award in 2019 among others.
%
%
\end{IEEEbiography}

\end{document}